\documentclass[]{amsart}
\usepackage{amsthm} 
\usepackage{latexsym} 
\usepackage{amsmath,amssymb}
\usepackage{verbatim}
\usepackage[dvips]{graphicx}
\newtheorem{theorem}{Theorem}
\newtheorem{remark}{Remark}
\newtheorem{proposition}{Proposition}

\newtheorem{example}{Example}
\newtheorem{corollary}{Corollary}
\newtheorem{lemma}{Lemma}
\newtheorem{definition}{Definition}

\begin{document}
\title[Arbitrage and Hedging in a Non Probabilistic Framework ]{Arbitrage and Hedging in a non probabilistic framework}
\author[Alvarez, Ferrando and Olivares]{A. Alvarez, S. Ferrando and P. Olivares\\
Department of Mathematics, Ryerson University}

\maketitle

\begin{abstract}
The paper studies the concepts of hedging and arbitrage in a non
probabilistic framework. It provides conditions for non
probabilistic arbitrage based on the topological structure of the
trajectory space and makes connections with the usual notion of
arbitrage. Several examples illustrate the non probabilistic arbitrage as well perfect replication of
options under continuous and discontinuous trajectories, the results
can then be applied in probabilistic models path by path. The
approach is related to recent financial models that go beyond
semimartingales, we remark on some of these connections and provide
applications of our results to some of these models.
\end{abstract}

\section{Introduction} \label{sec:introduction}

\noindent
Modern mathematical finance relies on the notions of arbitrage and hedging replication; generally,
these ideas are exclusively cast
in probabilistic frameworks.
The possibility of dispensing with probabilities  in mathematical finance, wherever possible, may have already occurred to several researchers.
A reason for this is that hedging results do not depend on the actual probability distributions but on the support of the probability measure.
Actually, hedging is clearly a pathwise notion and a simple view of arbitrage is that there is a portfolio that will no produce any
loss for all possible paths and there exists at least one path that will provide a profit. This
informal reasoning suggests that there is no need to use probabilities
to define the concepts even though probability has been traditionally used to do so.
The paper makes an attempt to study these two notions without probabilities in a direct and simple way.

From a technical point of view, we rely on a simple calculus for non differentiable
functions introduced in \cite{follmer} (see also \cite{sonderman}). This calculus is available
for a fairly large class of functions. Hedging results that only depend
on this pathwise calculus can be considered independently of probabilistic
assumptions. In \cite{bick} the authors take this point of view
and develop a discrete framework, and its associated limit, to
hedge continuous paths with a prescribed $2$-variation.
Reference \cite{bick} is mostly devoted to payoff replication for continuous trajectories and does not address the issue of non probabilistic arbitrage.
The present paper formally defines this last notion in a context that allows for trajectories with jumps,
develops some of the basic consequences and presents some simple applications including novel results to probabilistic frameworks.

Results on the existence of arbitrage in a
non standard framework (i.e. a non semimartingale price process) leads to interesting and challenging problems.
In order to gain a perspective on this issue, recall that a consequence of
the fundamental theorem of asset pricing of Delbaen and Schachermayer
in \cite{delbaen} is that under the NFLVR condition and considering simple
predictable portfolio strategies, the price process of the risky asset
necessarily must be a semi-martingale. Recent literature (\cite{cheridito}, \cite{jarrow}, \cite{russo}, \cite{valkeila-2}) describes pricing results in non
semi-martingale settings; the restriction
of the possible portfolio strategies has been a central issue in these works.  In \cite{cheridito} the permissible portfolio
strategies are restricted to those for which the time between consecutive
transactions is bounded below by some number $h$.  In \cite{jarrow}
these results are extended, also considering portfolios having a
minimal fixed time between successive trades. In \cite{russo} the notion
of $\mathcal{A}$-martingale is introduced in order to have the no-arbitrage
property for a given class $\mathcal{A}$ of admissible strategies.

We also treat this problem but with a different perspective, our
main object of study is a class of trajectories $\mathcal{J} = \mathcal{J}(x_0)$
starting at $x_0$. To this set of deterministic trajectories we associate a
class of admissible portfolios $\mathcal{A} = \mathcal{A}_{\mathcal{J}(x_0)}$ which, under some
conditions, is free of arbitrage and allow for perfect replication
to take place. These two notions, arbitrage and hedging, are defined
without probabilities. Once no arbitrage and hedging have been
established for the non probabilistic market model
($\mathcal{J}(x_0), \mathcal{A}_{\mathcal{J}(x_0)})$, these results could be used to
provide a fair price for the option being hedged. These pricing
results will not be stated explicitly in the paper and will be left
implicit.

Some technical aspects from our approach relate to the presentation in \cite{valkeila-2},
similarities with \cite{valkeila-2} are expressed mainly in the
use of a {\it continuity} argument which is also related to a {\it small ball}
property. Our approach eliminates the probability
structure altogether and replaces it with appropriate classes of trajectories;
the new framework also allows to accommodate continuous and discontinuous trajectories.

In summary, our work intends to develop a probability-free framework
that allows us to price by hedging and no-arbitrage. The results,
obtained under no probabilistic assumptions, will depend however, on
the topological structure of the possible trajectory space and a
restriction on the admissible portfolios by requiring a certain type
of continuity property. We connect our non probabilistic models with
stochastic models in a way that arbitrage results can be translated
from the non probabilistic models to stochastic models, even if
these models are not semi-martingales. The framework handles
naturally general subsets of the given trajectory space, this is not
the case in probabilistic frameworks that rely in incorporating
subsets of measure zero in the formalism.

The paper is organized as follows. Section
\ref{nonProbabilisticFramework} briefly introduces some of the
technical tools we need to perform integration with respect to
functions of finite quadratic variation and defines the basic
notions of the non probabilistic framework. Section
\ref{nonProbabilisticArbitrage} proves two theorems, they are key
technical results used throughout the rest of the paper. The
theorems provide a tool  connecting the usual notion of arbitrage
and non probabilistic arbitrage. Section
\ref{hedgingAndNoArbitrageInClasses}  introduces classes of modeling
trajectories, we prove a variety of non probabilistic hedging and no
arbitrage results for these classes. Section \ref{stoch-examples}
presents several examples in which we apply the non probabilistic
results to obtain new pricing results in several non standard
stochastic models. Appendix \ref{pathwiseIto} provides some
information on the analytical version of Ito formula that we rely
upon. Appendix \ref{technicalResults} presents some technical results needed in our developments.

\section{Non Probabilistic Framework}  \label{nonProbabilisticFramework}

We make use of the definition of integral with respect to functions with
unbounded variation but with finite quadratic variation  given in
\cite{follmer}.

Let $T>0$ be a fixed real number and let
$\mathcal{T} \equiv \left\{\tau_n\right\}_{n=0}^{\infty}$ where
\begin{equation} \nonumber
\tau_n=\left\{0=t_{n,0}< \cdots < t_{n,K(n)}=T\right\}
\end{equation}
are partitions of $[0,T]$ such that:
\begin{equation} \nonumber
\text{mesh}(\tau_n)=\max_{t_{n,k} \in \tau_n}|t_{n,k}-t_{n,k-1}| \to 0.
\end{equation}

Let $x$ be a real function on $[0, T]$ which is right
continuous and has left limits (RCLL for short), the space of such functions will be denoted by $\mathcal{D}[0,T]$. The following notation will be used, $\Delta
x_t=x_t-x_{t-}$ and $\Delta x_t^2=(\Delta
x_t)^2$.

Financial transactions will take place at times belonging to the above discrete grid but, otherwise, time will be treated continuously, in particular, the
values $x(t)$ could be observed in a continuous way.

A real valued RCLL function $x$ is of quadratic variation along $\mathcal{T}$
if the discrete measures
\begin{equation*}
\xi_n=\sum_{t_i \in \tau_n}(x_{t_{i+1}}-x_{t_i})^2\epsilon_{t_i}
\end{equation*}
converge weakly to a Radon measure $\xi$ on $[0, T]$ whose atomic part is
given by the quadratic jumps of $x$:
\begin{equation}  \label{generalDecomposition}
[x]_{t}^{\mathcal{T}}=\langle x \rangle_t^{\mathcal{T}}+\sum_{s \leq t}\Delta x_s^2,
\end{equation}
where $[x]^{\mathcal{T}}$ denotes the distribution function of
$\xi$ and $\langle x \rangle^{\mathcal{T}}$ its continuous part.

Considering $x$ as above and $y$ to be a function on $[0,T] \times
\mathcal{D}$, we will formally define the F\"ollmer's integral of $y$
respect to $x$ along $\tau$ over the interval $[0,t]$ for every $0
<t \leq T$. We should note that while the integral over $[0,t]$ for
$t<T$ will be defined in a proper sense, the integral over $[0,T]$
will be defined as an improper F\"ollmer's integral.

\begin{definition}\label{forward-integral}
Let $0<t<T$ and $x$ and $y$ as above, the F\"ollmer's integral of $y$
with respect to $x$ along $\mathcal{T}$ is given by
\begin{equation}\label{forward-int}
\int_0^t y(s,x) ~dx_s =\lim_{n \to \infty} \sum_{\tau_n \ni t_{n,i} \leq t}~ y({t_{n,i}}, x)~
(x(t_{n,i+1}) - x({t_{n,i}})),
\end{equation}
provided the limit in the right-hand side of (\ref{forward-int})
exists. The F\"ollmer's  integral over the whole interval $[0,T]$ is
defined in an improper sense:
\begin{equation} \nonumber
\int_0^T y(s, x)~ dx_s = \lim_{t \to T}  \int_0^t y(s, x)~ dx_s,
\end{equation}
provided the limit exists.
\end{definition}
Consider $\phi \in \mathcal{C}^1(\mathbb{R})$ (i.e. a function with
domain $\mathbb{R}$ and first derivative continuous), take $y(s, x)
\equiv \phi(x(s^-))$ then (\ref{forward-int}) exists. More
generally, if $y(s, x)= \phi(s, x(s^-), g_1(t,x^-), \ldots,
g_m(t,x^-))$ where the $g_i(\cdot,x)$ are functions of bounded
variation (that may depend on the past values of the trajectory $x$
up to time $t$) then (\ref{forward-int}) exists. Moreover, in these
two instances an Ito formula also holds, we refer to Appendix
\ref{pathwiseIto} for some details. Several of our non probabilistic
arbitrage arguments will depend only on assuming the {\it existence}
of integrals of the form $\int_0^t \phi(s, x) dx_s$ for a given
generic integrand $\phi(s, x)$ that (potentially) depends on all the
path values $x(t)$, $0 \leq t < s$, in these instances, and for the
sake of generality, we will work under this general assumption.

Next, we introduce the concepts of predictability, admissibility and
self-financing in a non probabilistic setting. The NP prefix will be
used throughout the paper indicating some non probabilistic concept.
For a given real number $x_0$, the central modeling object is a set of
trajectories $x$ starting at $x_0$, so $x$$:$$[0,T] \rightarrow \mathbb{R}$
with $x(0) = x_0$. We will assume that these functions are RCLL and
belong to a given {\it set of trajectories} $\mathcal{J}(x_0)$. In order
to easy the notation, this last class may be written as $\mathcal{J}$
when the point $x_0$ is clear from the context.

Some of our results apply to rather general trajectory classes, particular trajectory classes
will be needed to deal with hedging results and applications to classical models
and will be introduced at due time.

We assume the existence of a non risky asset with interest
rates $r \geq 0$ which, for simplicity, we will assume constant,
and a risky asset whose price trajectory belongs to a function class $\mathcal{J}(x_0)$.
For convenience, in several occasions, we will restrict our arguments to the case $r=0$.

\vspace{.1in}
A NP-portfolio $\Phi$ is a function $\Phi$$:$ $[0, T] \times
\mathcal{J}(x_0) \rightarrow \mathbb{R}^2$, $\Phi =(\psi, \phi)$,
satisfying $\Phi(0,x)=\Phi(0,x')$ for all $x, x' \in \mathcal{J}(x_0)$.
We will also consider the associated projection functions $\Phi_x$$:$$ [0,T]
\rightarrow \mathbb{R}^2$ and $\Phi_t$$:$$ \mathcal{J}(x_0) \rightarrow
\mathbb{R}^2$, for fixed $x$ and $t$ respectively.

The value of a NP-portfolio $\Phi$ is the function $V_{\Phi}$$:$$
[0,T]\times \mathcal{J}(x_0) \rightarrow \mathbb{R}$ given by:
\begin{equation}  \nonumber 
V_{\Phi}(t,x) \equiv \psi(t, x)+ \phi(t,x) ~x(t).
\end{equation}

\begin{definition} \label{def:NP-conditions}
Consider a class $\mathcal{J}(x_0)$ of trajectories starting at $x_0$:

\begin{itemize}
\item [i)] A portfolio $\Phi$ is said to be NP-predictable
if $\Phi_t(x) = \Phi_t(x')$ for all $x, x' \in \mathcal{J}(x_0)$
such that $x(s) = x'(s)$ for all $0 \leq s < t$ and $\Phi_x(\cdot)$
is a left continuous function with right limits (LCRL functions for
short) for all $x \in \mathcal{J}(x_0)$.
\item [ii)] A portfolio $\Phi$ is said to be NP-self-financing if the integrals $\int_0^t \psi(s,x)~ ds$
and $\int_0^t \phi(s,x)~ d x_s$ exist for all $x \in
\mathcal{J}(x_0)$ as a Stieljes and F\"ollmer integrals
respectively, and
\begin{equation} \nonumber
V_{\Phi}(t,x)=V_0 + \int_0^t  \psi(s, x) ~r~ds + \int_0^t \phi(s, x) d x_s, \; \forall x \in \mathcal{J}(x_0),
\end{equation}
where $V_0=V(0,x)=\psi(0,x) + \phi(0,x) ~~x(0)$ for any $x \in \mathcal{J}(x_0)$.
\item [iii)] A portfolio $\Phi$ is said to be NP-admissible if $\Phi$ is NP-predictable,
NP-self-financing and $V_{\Phi}(t,x) \geq - A $, for a constant $A = A(\Phi) \geq 0$, for all $t \in [0, T]$ and all $x \in \mathcal{J}(x_0)$.
\end{itemize}
\end{definition}
\begin{remark}
\noindent
\begin{enumerate}
\item Two identical trajectories up to time $t$ will lead to identical portfolio strategies
up to time $t$.
\item Notice that the notion of NP-admissible portfolio is relative to a given class of trajectories $\mathcal{J}$,
classes of NP portfolios will be denoted  $\mathcal{A}_{\mathcal{J}}$ or $\mathcal{A}$ for simplicity .
\end{enumerate}
\end{remark}

The following definition is central to our approach.

\begin{definition}\label{def:NP-market}
A NP-market is a pair $(\mathcal{J},\mathcal{A})$ where $\mathcal{J}$ represents a class of
possible trajectories for a risky asset and
 $\mathcal{A}$ is an admissible class of portfolios .
\end{definition}

The following definition provides the notion of arbitrage in a non probabilistic framework.
\begin{definition} \label{nPArbitrage}
We will say that there exists NP-arbitrage in the NP-market $(\mathcal{J},\mathcal{A})$ if there
exists a portfolio $\Phi \in \mathcal{A}$
such that $V_{\Phi}(0,x)=0$ and $V_{\Phi}(T,x) \geq 0$  for all $x \in \mathcal{J}$,
and there exists at least one
trajectory $x^* \in \mathcal{J}$ such that $V_{\Phi}(T,x^*) > 0$. If no NP-arbitrage exists
then we will say that the NP-market $(\mathcal{J},\mathcal{A})$ is NP-arbitrage-free.
\end{definition}

The notion of probabilistic market that we use through the paper
is similar to the one in \cite{valkeila-2}.
Assume a filtered probability space $(\Omega, \mathcal{F},
(\mathcal{F}_t)_{0 \leq t \leq T}, \mathbb{P})$ is given. Let $Z$ be an
adapted stochastic process modeling asset prices defined on this
space.

A portfolio strategy $\Phi^z$ is a
pair of stochastic processes $\Phi^z=(\psi^z, \phi^z)$. The value of
a portfolio $\Phi^z$ at time $t$ is a random variable given by:
\[
V_{\Phi^z}(t)=\psi^z_t+ \phi^z_t Z_t.
\]

A portfolio $\Phi^z$ is self-financing if the integrals $\int_0^t
\psi^z_s(\omega) ds$ and $\int_0^t \phi^z_s(\omega) d Z_s(\omega)$
exist $\mathbb{P}$~-a.s. as a Stieltjes integral and a F\"ollmer stochastic integral
respectively and
\begin{equation} \nonumber
V_{\Phi^z}(t)=V_{\Phi^z}(0)+r \int_0^t  \psi^z_s ds + \int_0^t
\phi^z_s d Z_s, \; \mathbb{P}~-a.s.
\end{equation}
 A portfolio $\Phi^z$ is admissible if $\Phi^z$ is self-financing, predictable (i.e. measurable with respect to $\mathcal{F}_{t-}$)
 and there exists
 $A^z = A^z(\Phi^z)  \geq 0$ such that $V_{\Phi^z}(t) \geq -~A^z$ $\mathbb{P}$~-a.s. $\forall ~t \in [0,T]$.

\begin{definition}\label{stoch-market}
A stochastic market defined on a filtered probability space
$(\Omega, \mathcal{F}, (\mathcal{F}_t)_{t \geq 0}, \mathbb{P})$ is a
pair $(Z,\mathcal{A}^Z)$ where $Z$ is an adapted stochastic process
modeling asset prices and $\mathcal{A}^Z$ is a class of admissible
portfolio strategies.
\end{definition}
\begin{remark}
We assume $\mathcal{F}_0$ is the trivial sigma algebra, furthermore, without loss of generality, we will assume
that the constant $z_0 = Z(0, w)$ is fixed, i.e. we assume the same initial value for all paths.
The constant $V_{\Phi^z}(0, w)$ will also be denoted $V_{\Phi^z}(0, z_0)$.
\end{remark}

\noindent The notion of arbitrage in a probabilistic market is the classical notion
of arbitrage (which in this paper will be referred simply as {\it arbitrage}).
The market $(Z,\mathcal{A}^Z)$ defined over $(\Omega, \mathcal{F},
(\mathcal{F}_t)_{t \geq 0}, \mathbb{P})$ has arbitrage opportunities
if there exists $\Phi^z \in \mathcal{A}^Z$ such that
$V_{\Phi^z}(0)=0$ and $V_{\Phi^z}(T) \geq 0$  $\mathbb{P}$~-a.s.,
and $\mathbb{P}(V_{\Phi^z}(T)>0)>0$.

\section{Non Probabilistic Arbitrage Results} \label{nonProbabilisticArbitrage}

Our technical approach to establish NP arbitrage results is to link them to
classical arbitrage results. Somehow surprisingly, this connection will allow
us to apply the so obtained NP results to prove non existence of arbitrage results in new probabilistic markets.
This section provides two basic theorems, Theorem \ref{main-no-arbitrage}
allows to construct NP markets free of arbitrage from a given
arbitrage free probabilistic market. Applications of this theorem are given
in Section \ref{hedgingAndNoArbitrageInClasses}. Theorem \ref{main-no-arbitrageDual}
presents a dual result allowing to construct probabilistic, arbitrage free,
markets from a given NP market which is arbitrage free. Applications of this
theorem are given in Section \ref{stoch-examples}.

In order to avoid repetition we will make the following standing
assumption for the rest of the section: for all the set of
trajectories $\mathcal{J}$ and price processes $Z$ to be considered,
there exists a metric space $(\mathcal{S}, d)$ satisfying
$\mathcal{J} \subseteq \mathcal{S}$ and $Z(\Omega) \subseteq
\mathcal{S}$ up to  a set of measure zero. All topological notions considered in the paper are relative to this metric space. Examples in later
sections will use the uniform distance $d(x,y)=
||x-y||_{\infty}$ where $||x||_{\infty} \equiv \sup_{s \in [0,T]}
|x(s)|$ and $\mathcal{S}$ the set of continuous functions $x$ with
$x(0) = x_0 = z_0$. For trajectories with jumps, later sections will
use the Skorohod distance, denoted by $d_s$,  and $\mathcal{S}$ the set of RCLL
functions $x$  with $x(0) = x_0 = z_0$.

While the main results in this section can be formulated in terms
of isomorphic and V-continuous portfolios (see Definitions \ref{isomorphic} and \ref{vContinuity}),
the presentation makes use the
following weaker notion of connected portfolios; this approach provides
stronger results.

\begin{definition} \label{npConnectedToProbabilistic}
Let $(\mathcal{J},\mathcal{A})$ and $(Z, \mathcal{A}^Z)$ be
respectively NP and stochastic markets. $\Phi \in \mathcal{A}$ is
said to be connected to $\Phi^z \in \mathcal{A}^Z$ if the following
holds in a set of full measure:
\begin{equation} \nonumber
V_{\Phi^z}(0,z_0)=V_{\Phi}(0, x_0)
\end{equation}
and for any fixed $x \in \mathcal{J}$ and arbitrary $\rho
>0$ there exists $\delta = \delta(x, \rho) >0$ such that
\begin{equation}  \label{weakLowerContinuity}
\mbox{if}~~d(Z(w), x) < \delta~\mbox{then}~~V_{\Phi^z}(T,\omega)
\geq V_{\Phi}(T, x)- \rho.
\end{equation}
\end{definition}

Given a class of stochastic portfolios $\mathcal{A}^Z$, Section
\ref{hedgingAndNoArbitrageInClasses} constructs NP-admissible
portfolios $\Phi \in \mathcal{A}$ with the goal of obtaining NP arbitrage
free markets $(\mathcal{J}, \mathcal{A})$. Each such collection of
portfolios  $\mathcal{A}$ is defined as the largest class of NP
admissible portfolios connected to an element from $\mathcal{A}^Z$.
Here is the required definition.

\begin{definition}\label{def:connection}
Let $(Z,\mathcal{A}^Z)$ be a stochastic market on $(\Omega,
\mathcal{F}, \mathcal{F}_t, P)$, and $\mathcal{J}$  a set of
trajectories. Define:
\begin{equation} \nonumber
[\mathcal{A}^Z] \equiv \{\Phi: \Phi ~\mbox{is NP-admissible},~~
~~\exists~~  \Phi^z \in \mathcal{A}^Z ~\mbox{s.t.}~~\Phi~~~~\mbox{is
connected to} ~~\Phi^z\}.
\end{equation}
\end{definition}

\begin{theorem}   \label{main-no-arbitrage}
Let $(Z, \mathcal{A}^Z)$ be a stochastic market and $\mathcal{J}$  a
set of trajectories. Furthermore, assume the following conditions are satisfied:

\vspace{.05in} \noindent $C_0:$~~~~ $Z(\omega) \subseteq
\mathcal{J}$~a.s.

\noindent
$C_1:$~~~~$Z$ satisfies a small ball property with respect to
the metric $d$ and the space $\mathcal{J}$, namely $\mbox{for all}~~~ \epsilon>0$:
\[
P\left(d(Z,x)<\epsilon \right)>0, \; \forall x \in \mathcal{J}.
\]

\vspace{.05in} \noindent
Then, the following statement holds.

\vspace{.05in} \noindent If $(Z, \mathcal{A}^Z)$ is arbitrage free
then $(\mathcal{J}, [\mathcal{A}^Z])$ is NP-arbitrage free.
\end{theorem}
\begin{proof}
We proceed to prove the statement by contradiction. Suppose then,
that there exists a NP-arbitrage portfolio $\Phi \in
[\mathcal{A}^z]$; therefore $V_{\Phi}(0,x) =0$ and $V_{\Phi}(T,x)
\geq 0$ for all $ x \in \mathcal{J}$ and there is also $x^{\ast} \in
\mathcal{J}$ satisfying $V_{\Phi}(T,x^{\ast}) > 0$. From the
definition of $[\mathcal{A}^z]$, it follows that there exists
$\Phi^z \in \mathcal{A}^z$ connected to $\Phi$ satisfying
$V_{\Phi^z}(0, z_0) = V_{\Phi}(0, x_0) =0$. Using $C_0$, consider the case
when there exist $\hat{\Omega}$, a measurable set of full measure,
such that $Z(\omega) \subseteq \mathcal{J}$ hods for all $w \in
\hat{\Omega}$. Assume further, there exists a measurable set
$\hat{\Omega}_1 \subseteq \hat{\Omega}$ with $P(\hat{\Omega}_1) >0$
such that $V_{\Phi^z}(T, \omega) < 0$ holds for all $\omega \in
\hat{\Omega}_1$. The relation ``$\Phi$ is connected to $\Phi^z$''
holds in a set of full measure which is independent on any given
$x$, then, we may assume without loss of generality that
(\ref{weakLowerContinuity}) holds for all $w \in \hat{\Omega}_1$.
Consider an arbitrary $\hat{\omega} \in \hat{\Omega}_1$ and use the
notation $x \equiv Z(\hat{\omega}) \in \mathcal{J}$; for an
arbitrary $\rho >0$ we then have: $V_{\Phi^z}(T, \hat{\omega}) \geq
V_{\Phi}(T, x) - \rho$; $\rho$ being arbitrary, this gives a
contradiction. Therefore, $V_{\Phi^z}(T, \omega) \geq 0$ a.s. holds.
Consider now $x^{\ast}$ fixed as above, an arbitrary  $\rho >0 $ and
$\delta > 0$ given by the fact that $\Phi$ is connected to $\Phi^z$.
Condition $C_1$ implies that the set $B_{\rho} = \{w: d(Z(w),
x^{\ast}) < \delta \}$ satisfies $P(B_{\rho}) >0$ for any $\rho >0$,
then using (\ref{weakLowerContinuity}) we obtain $V_{\Phi^z}(T,
\omega) \geq V_{\Phi}(T, x^{\ast})- \rho$, which we may assume
without loss of generality holds for all $w \in B_{\rho}$. Clearly,
$V_{\Phi}(T, x^{\ast}) >0$ being fixed, there exist a small
$\rho^{\ast} >0$ such that $V_{\Phi^z}(T, \omega) > 0$ for all $w
\in B_{\rho^{\ast}}$ and $P(B_{\rho^{\ast}})
>0$. This concludes the proof.
\end{proof}

In Section \ref{hedgingAndNoArbitrageInClasses} we are faced with
the following problem: given $\mathcal{A}^z$, we need to prove that
a given NP admissible portfolio $\Phi$ belongs to $[\mathcal{A}^z]$.
The following proposition provides a sufficient condition to check that
a given $\Phi$ is connected to a certain $\Phi^z$. The stronger setting
of the proposition also allows to see the condition (\ref{weakLowerContinuity}) as a weak
form of lower semi continuity of the value of the NP portfolio.

\begin{proposition}  \label{sufficientConditionForNPPortfolio}
Let $(Z, \mathcal{A}^Z)$ be a stochastic market and $\mathcal{J}$  a
set of trajectories and assume that $C_0$ from Theorem \ref{main-no-arbitrage} holds.
Then, if a NP-admissible portfolio $\Phi$ is such
that $V_{\Phi}(T, \cdot)$$:$$ \mathcal{J} \rightarrow \mathbb{R}$ is
lower semi-continuous with respect to metric $d$ and there exist
$\Phi^z \in \mathcal{A}^Z$ such that $~V_{\Phi^z}(0,z_0)=V_{\Phi}(0, x_0)$ and $V_{\Phi^z}(T,w)
= V_{\Phi}(T,Z(w))$   then $\Phi$ is connected to $\Phi^z$ (so
$\Phi \in [\mathcal{A}^Z]$.)
\end{proposition}
\begin{proof}
Consider $\Phi$ and $\Phi^z$ satisfying the hypothesis of the
proposition. The lower semi continuity means that for a given $x \in
\mathcal{J}$ and any  $\rho >0$ there exists $\delta >0$ satisfying:
if $d(x', x) < \delta$, with $x' \in \mathcal{J}$ then
\begin{equation} \label{upperSemiContinuous}
V_{\Phi}(T, x') \geq V_{\Phi}(T, x) - \rho.
\end{equation}
Consider now $w$ to be in the set of full measure where   $Z(\Omega)
\subseteq \mathcal{J}$ holds; fix  $x \in \mathcal{J}$ and $\rho >
0$ arbitrary. Consider now $\delta$ as given by the lower semi
continuity assumption, then, if $d(Z(w), x) < \delta$, taking $x'
\equiv Z(w)$ we obtain $V_{\Phi^z}(T,w) = V_{\Phi}(T,x') \geq
V_{\Phi}(T,x) - \rho$.
\end{proof}

In order to construct arbitrage free probabilistic markets from NP
markets free of arbitrage we will make use of the following notion.

\begin{definition} \label{probabilisticConnectedToNP}
Let $(\mathcal{J},\mathcal{A})$ and $(Z, \mathcal{A}^Z)$ be
respectively NP and stochastic markets. $\Phi^z \in \mathcal{A}^Z$
is said to be connected to $\Phi \in \mathcal{A}$ if the following
holds in a set of full measure:
\begin{equation} \nonumber
V_{\Phi^z}(0,z_0)=V_{\Phi}(0, x_0)
\end{equation}
and for any fixed $x \in \mathcal{J}$ and arbitrary $\rho
>0$ there exists $\delta = \delta(x, \rho) >0$ such that
\begin{equation}  \label{weakUpperContinuity}
\mbox{if}~~d(Z(w), x) < \delta~\mbox{then}~~V_{\Phi}(T,x) \geq
V_{\Phi^z}(T, \omega)- \rho.
\end{equation}
\end{definition}

In order to apply results obtained  for NP-markets to stochastic
markets, in particular non-semimartingale processes, Section
\ref{stoch-examples} makes use of the following construction:
starting from a class of NP admissible portfolios $\mathcal{A}$, a
class of portfolios $\mathcal{A}^Z$ is defined as the largest
collection of admissible portfolios  which are connected to elements
from $\mathcal{A}$. This construction will give an arbitrage free
stochastic market $(Z, \mathcal{A}^Z)$. Here is the required
definition.

\begin{definition} \label{def:connectionII}
Let $(\mathcal{J},\mathcal{A})$ be a NP market and  $(\Omega,
\mathcal{F}, (\mathcal{F}_t)_{t \geq 0}, P)$ a filtered probability
space. Let $Z$ be an adapted stochastic process defined on this
space. Define:
\begin{equation} \nonumber
[\mathcal{A}]^Z \equiv \{\Phi^z: \Phi^z ~\mbox{is admissible},~~
~~\exists~~  \Phi \in \mathcal{A} ~\mbox{s.t.}~~\Phi^z~~~~\mbox{is
connected to} ~~\Phi\}.
\end{equation}
\end{definition}

The following theorem is the dual version of Theorem \ref{main-no-arbitrage}.
\begin{theorem}   \label{main-no-arbitrageDual}
Let $(\mathcal{J},\mathcal{A})$ be a NP market and  $(\Omega,
\mathcal{F}, (\mathcal{F}_t)_{t \geq 0}, P)$ a filtered probability
space. Let $Z$ be an adapted stochastic process defined on this
space. Furthermore, assume $C_0$ and $C_1$ from Theorem \ref{main-no-arbitrage}
 hold.

\vspace{.05in} \noindent Then the following statement holds:

\vspace{.05in} \noindent If $(\mathcal{J},\mathcal{A})$ is
NP-arbitrage free
then $(Z, [\mathcal{A}]^z)$ is arbitrage free.
\end{theorem}
\begin{proof}
We argue by contradiction, suppose there exists an arbitrage portfolio $\Phi^z \in
[\mathcal{A}]^z$; therefore, $V_{\Phi^z}(0,w) =0$ and  $V_{\Phi^z}(T,w)
\geq 0$~~a.s. Moreover,
there exists a measurable set $D \subseteq \Omega$ satisfying
 \begin{equation}  \label{arbitragePaths}
 V_{\Phi^z}(T,w) > 0~~\mbox{for all}~~w \in D ~\mbox{and}~~  P(D)> 0.
\end{equation}
Because $\Phi^z \in [\mathcal{A}]^z$, we know that $\Phi^z$ is
connected to some $\Phi \in \mathcal{A}$. Then, $0=V_{\Phi^z}(0,
z_0) = V_{\Phi}(0, x_0) = V_{\Phi}(0, x)$ for all $x \in
\mathcal{J}$. Assume now there exists $\tilde{x} \in \mathcal{J}$
and $V_{\Phi}(T, \tilde{x}) < 0$, by $C_1$ and
(\ref{weakUpperContinuity}), given $\rho
>0$, we obtain
\begin{equation}
V_{\Phi}(T, \tilde{x}) \geq V_{\Phi^z}(T, \omega) - \rho
\end{equation}
a.s. for $w \in B_{\rho} \equiv \{\omega: d(Z(\omega), \tilde{x}) <
\delta \}$ and $\delta >0$ is as in (\ref{weakUpperContinuity}).
Using the fact that $V_{\Phi^z}(T, \omega) \geq 0$ we arrive at a
contradiction and conclude $V_{\Phi}(T, x) \geq  0$ for all $x \in
\mathcal{J}$. Assume now $V_{\Phi}(T, x) =  0$ for all $x \in
\mathcal{J}$, because $C_0'$ and (\ref{arbitragePaths}), there
exists $\omega^{\ast} \in D$ and $x^{\ast} \equiv  Z(\omega^{\ast})
\in \mathcal{J}$. The relation ``$\Phi^z$ is connected to $\Phi$''
holds in a set of full measure which is independent on any given
$x$, then, we may assume without loss of generality that
(\ref{weakUpperContinuity}) holds for $\omega^{\ast}$. Then, using
$C_1$ we obtain: $V_{\Phi}(T, x^{\ast}) \geq V_{\Phi^z}(T,
\omega^{\ast}) - \rho$ for all $\rho >0$. This implies $V_{\Phi}(T,
x^{\ast})  > 0$.
\end{proof}

The following proposition provides sufficient conditions to check that a certain $\Phi^z$
is connected to a NP portfolio $\Phi$.

\begin{proposition}  \label{sufficientConditionForProbabilisticPortfolio}
Let $(\mathcal{J},\mathcal{A})$ be a NP market,
$Z$ an adapted stochastic process defined on the
filtered probability space $(\Omega,\mathcal{F}, (\mathcal{F}_t)_{t \geq 0}, P)$
and assume $C_0$ from Theorem \ref{main-no-arbitrage} holds.
Then, if $\Phi^z$ is an
admissible portfolio such that there exists $\Phi \in \mathcal{A}$ satisfying:
$V_{\Phi}(T, \cdot): \mathcal{J} \rightarrow \mathbb{R}$ is upper
semi-continuous with respect to metric $d$, $~V_{\Phi^z}(0,z_0)=V_{\Phi}(0, x_0)$ and
$V_{\Phi^z}(T,w) = V_{\Phi}(T,Z(w))$ a.s.,  then $\Phi^z$ is
connected to $\Phi$ (so $\Phi^z \in [\mathcal{A}]^Z$.)
\end{proposition}
\begin{proof}
Consider $\Phi$ and $\Phi^z$ satisfying the hypothesis of the
proposition. The upper semi continuity means that for a given $x \in
\mathcal{J}$ and any  $\rho >0$ there exists $\delta >0$ satisfying:
if $d(x', x) < \delta$, with $x' \in \mathcal{J}$ then
\begin{equation} \nonumber 
V_{\Phi}(T, x) \geq V_{\Phi}(T, x') - \rho.
\end{equation}
Consider now $w$ to be in the set of full measure where   $Z(\Omega)
\subseteq \mathcal{J}$ holds; fix  $x \in \mathcal{J}$ and $\rho >
0$ arbitrary. Consider now $\delta$ as given by the upper semi
continuity assumption, then, if $d(Z(w), x) < \delta$, taking $x'
\equiv Z(w)$ we obtain $ V_{\Phi}(T,x) \geq V_{\Phi}(T,x') - \rho =
V_{\Phi^z}(T,w)  - \rho$.
\end{proof}

\vspace{.2in}
For simplicity, in most of our further developments, we will make use of stronger notions
than connected and lower and upper semi-continuous portfolios. Namely, isomorphic and V-continuous
portfolios, here are the definitions.

\begin{definition} \label{isomorphic}
Let $(\mathcal{J},\mathcal{A})$ and $(Z, \mathcal{A}^Z)$ be
respectively NP and stochastic markets and assume the condition $C_0$ from
Theorem \ref{main-no-arbitrage} holds.
A NP portfolio $\Phi \in \mathcal{A}$ and $\Phi^z \in \mathcal{A}^Z$
are said to be isomorphic if $\mathbb{P}$-a.s.:
\[
\Phi^z(t,\omega)=\Phi(t,Z(\omega))
\]
for all $0 \leq t \leq T$.
\end{definition}

\begin{definition} \label{vContinuity}
Let $(\mathcal{J},\mathcal{A})$ be a NP market. A NP portfolio
$\Phi \in \mathcal{A}$ is said to be V-continuous with respect to $d$
if the functional $V_{\Phi}(T,\cdot)$$:$$ \mathcal{J} \rightarrow \mathbb{R}$
is continuous with respect to the topology induced on $\mathcal{J}$ by
distance $d$.
\end{definition}

Whenever the distance $d$ is understood from the context we will
only refer to the portfolio as {\it V-continuous}. The intuitive
notion of a V-continuous portfolio is that small changes in the asset
price trajectory will lead to small changes to the final value of
the portfolio.

\begin{remark}
Clearly, if $\Phi_T$$:$$ \mathcal{J} \rightarrow \mathbb{R}^2$ is
continuous then $\Phi$ is V-continuous.
\end{remark}

Propositions \ref{sufficientConditionForNPPortfolio}
and \ref{sufficientConditionForProbabilisticPortfolio} plus
the definition of V-continuity give the following corollary.

\begin{corollary} \label{isomophic-v-cont}
Consider the setup of Definition \ref{isomorphic}. In particular, consider
$\Phi$ and $\Phi^z$ to be isomorphic, furthermore, assume $\Phi$ to be V-continous, then:
\begin{itemize}
\item $\Phi$ is connected to  $\Phi^z$ and so $\Phi \in [\mathcal{A}^Z]$.
\item $\Phi^z$ is connected to  $\Phi$ and so $\Phi^z \in [\mathcal{A}]^Z$.
\end{itemize}
\end{corollary}

In many of our examples, we will rely on  Corollary  \ref{isomophic-v-cont}
to check if given portfolios belong to  $[\mathcal{A}^Z]$ or $[\mathcal{A}]^Z$.
In each of our examples, introduced in later sections, it will arise
the question on how large are the classes of portfolios
$[\mathcal{A}^Z]$ and $[\mathcal{A}]^Z$ as the Definitions
\ref{def:connection} and \ref{def:connectionII} do not provide a
direct characterization of its elements. For each of our examples we
will prove that specific classes of portfolios do belong to
$[\mathcal{A}^Z]$ and $[\mathcal{A}]^Z$, answering the question in
general is left to future research.

\vspace{.2in}
\noindent \textbf{Arbitrage in subsets}\\
\\
Consider $\mathcal{J}^{\ast} \subseteq \mathcal{J}$, it is natural to look for conditions
that provide a relationship between the arbitrage opportunities of these two sets.
The NP framework allows a simple result, Proposition \ref{subclasses} below, which provides
a clear contrast with the probabilistic framework
which, in particular, is not able to provide an answer when $\mathcal{J}^{\ast}$ is a subset of measure zero (see Example \ref{rational-end-points}). There exist  cases, for
example if $\mathcal{J}^{\ast}=\{x \in  \mathcal{J}: x_T> x_0e^{rT}\}$,
for which there exists an obvious NP-arbitrage portfolio by borrowing money from the bank
and investing on the asset. Proposition
\ref{subclasses} shows that the no-arbitrage property for a NP-market $(\mathcal{J},\mathcal{A})$ is inherited by
NP-markets whose trajectories $\mathcal{J}^{\ast}$ are dense on $\mathcal{J}$.

\begin{proposition} \label{subclasses}
Consider the NP-market $(\mathcal{J},\mathcal{A})$ where $\mathcal{A}$ is some class
of NP-admissible, V-continuous portfolio strategies (with respect to metric $d$). Let $\mathcal{J}^{\ast} \subset \mathcal{J}$
be a subclass of trajectories such that $\mathcal{J}^{\ast}$ is dense in $\mathcal{J}$ with respect to the metric
$d$ and consider $\mathcal{A}_{\mathcal{J}^{\ast}}$ to be the restriction of portfolio strategies in $\mathcal{A}$
to the subclass $\mathcal{J}^{\ast}$. Then the NP-Market $(\mathcal{J}^{\ast},\mathcal{A}_{\mathcal{J}^{\ast}})$ is
NP-arbitrage-free if $(\mathcal{J},\mathcal{A})$ is NP-arbitrage-free.
\end{proposition}
\begin{proof}
Assuming that there exists $\Phi \in \mathcal{A}_{\mathcal{J}^{\ast}}$, an arbitrage opportunity on $(\mathcal{J}^{\ast},\mathcal{A}_{\mathcal{J}^{\ast}})$,  we will derive
an arbitrage strategy in $(\mathcal{J},\mathcal{A})$. To achieve this end, it is enough to prove that
$V_{\Phi}(0,x)= 0$ and $V_{\Phi}(T,x) \geq  0$, both relations valid for all $x \in \mathcal{J}$. The first
relation should hold because of the density assumption and the fact that $V_{\Phi}(0, \cdot)$ is a
continuous function on $\mathcal{J}$. The second relationship follows similarly using
continuity of $V_{\Phi}(T, \cdot)$ on $\mathcal{J}$.
\end{proof}

\section{Examples: Arbitrage and Hedging in Trajectory Classes}  \label{hedgingAndNoArbitrageInClasses}

This section provides examples of NP markets $(\mathcal{J},
\mathcal{A})$ which are free of arbitrage and in which general
classes of payoffs can be hedged. Several of the results from
Section \ref{nonProbabilisticArbitrage} are applied in order to gain
a more complete understanding of these examples, in particular, we
provide several details about the characterizations of the
portfolios $\Phi \in \mathcal{A}$.

A main example deals with continuous trajectories (this set is
denoted by $\mathcal{J}_{\tau}^{\sigma}$), other examples  deal with
trajectories containing jumps. Several aspects of these different
examples are treated in a uniform way illustrating the flexibility
of using different topologies in the trajectory space. The classes
of trajectories to be introduced could be considerably enlarged by
allowing the parameter $\sigma$ to be a function of $t$ (obeying
some regularity conditions). Our results apply to such (extended)
classes as well, in the present paper we will restrict $\sigma$ to
be a constant for simplicity. We also restrict to hedging results to
path independent derivatives but expect the results can be extended
to path independent derivatives as well.

The replicating portfolio strategies that we will obtain in a NP-market are essentially the same that in the corresponding stochastic frameworks,
for example, to replicate a payoff when prices lie in our example
$\mathcal{J}_{\tau}^{\sigma}$, comprised of continuous trajectories,
we use the well known delta-hedging as in the Black-Scholes model.
In the available literature there exist several results related to
the robustness of delta hedging, see for example: \cite{bick},
\cite{Kloeden}, \cite{valkeila-2} and \cite{russo}. A point to
emphasize is the fact that the replication results, being valid in a
different sense (probability-free), are valid also when considering
subclasses of trajectories $\mathcal{J}^{\ast} \subset \mathcal{J}$.
Formally, this fact is not available in probabilistic frameworks
due to the technical reliance on sets of measure zero or non-measurable
sets.

\subsection{Non Probabilistic Black-Scholes Model}\label{ssect:cont}

Denote by $\mathcal{Z}_{\mathcal{T}}([0, T])$ the collection of
all continuous functions $z(t)$ such that
$\left[ z \right]_{t}^{\mathcal{T}} = t$ for $0\leq t \leq T$ and
$z(0) =0$. Notice that $\mathcal{Z}_{\mathcal{T}}([0, T])$ includes a.s. paths of
Brownian motion if $\mathcal{T}$ is a refining sequence of
partitions (\cite{klein}.)

\noindent For a given  sequence of subdivisions $\mathcal{T}$
define,
\begin{itemize}
\item Given constants $\sigma>0$ and $x_0 > 0$,  let $ \mathcal{J}_{\tau}^{\sigma}(x_0)$ to be the class of
all real valued functions $x$ for which there exists $z \in \mathcal{Z}_{\mathcal{T}}([0, T])$
such that:
\begin{equation} \nonumber 
x(t) = x_{0}~e^{\sigma z(t)}.
\end{equation}
\end{itemize}
According to (\ref{quadVarForComposition}), the class $ \mathcal{J}^{\sigma}_{\tau}(x_0)$ is the class of
continuous functions $x$ with $x(0) = x_0$ and quadratic variation
satisfying $d\langle x \rangle^{\tau}_t = \sigma^2 x(t)^2 dt$. The trajectory class $ \mathcal{J}^{\sigma}_{\tau}(x_0)$ will be considered as a subset of the continuous functions with the uniform topology induced by the uniform distance.

\begin{remark} \label{generalizedClasses}
The class $ \mathcal{J}_{\tau}^{\sigma}(x_0)$ includes
trajectories of processes different than the geometric Brownian motion,
as an illustration we indicate
that if $z= B + y$, with $B$ a Brownian motion and
$y$ a process with zero quadratic variation, then the trajectories of
$z$ belongs to
$\mathcal{Z}_{\tau}([0, T])$, hence the trajectories of the process
$x_0~e^{\sigma z}$ belong to $ \mathcal{J}_{\tau}^{\sigma}(x_0)$.
\end{remark}

As previously suggested, the hedging results in this class have
been already obtained, more or less explicitly, in several papers (see \cite{bick} and \cite{Kloeden}
for example).

\begin{theorem}\label{Non-stoch-BS}
Let $\mathcal{J}^{\ast}$ be a class of possible trajectories,
$\mathcal{J}^{\ast} \subset \mathcal{J}_{\tau}^{\sigma}$ and let
$v(\cdot,\cdot): [0,T] \times \mathbb{R^+} \rightarrow \mathbb{R}$
be the solution of the PDE
\begin{equation}\label{BS-PDE}
\frac{\partial v}{\partial t}(t,x) + r ~x \frac{\partial v}{\partial x}(t,x) +
\frac{\sigma^2 x^2}{2}\frac{\partial^2 v}{\partial x^2}(t,x)-r~v(t,x) = 0
\end{equation}
with terminal condition $v(T,x)=h(x)$ where $h(\cdot)$ is
Lipschitz. Then, the delta hedging NP-portfolio $\Phi_t  \equiv
(v(t,x(t))-\varrho(t,x(t))x(t), \varrho(t,x(t)))$ where
$\varrho(t,x)  \equiv \frac{\partial v}{\partial x}(t,x)$
replicates the payoff $h$ at maturity time $T$ for all $x \in
\mathcal{J}^{\ast}$.
\end{theorem}
\begin{proof}
The existence and uniqueness of the solution of (\ref{BS-PDE})
is guaranteed because $h$ is Lipschitz.
If $x \in \mathcal{J}^{\ast} \subset \mathcal{J}_{\tau}^{\sigma}(x_0)$ then we know
that $x$ is of quadratic variations and
$d \langle x \rangle_t^{\tau}= \sigma^2 x(t)^2  dt$, so applying It\^{o}-F\"ollmer
formula, taking $\varrho(t,x)=\frac{\partial v}{\partial x}(t,x)$, using (\ref{BS-PDE})
and noticing that the integral $\int_0^T r(v(s,x(s))-\varrho(s,(s,x(s))x(s))ds$
exists, we obtain:
\begin{equation} \label{eq:delta-hedging}
h(x(T))= v(T, x(T))= \lim_{u \to T} v(u,x(u))=
\end{equation}
\begin{equation*}
\lim_{u \to T} \left[ v(0,x(0))+ \int_0^u \frac{\partial v}{\partial t}(s,x(s))ds
 + \frac{1}{2}\int_0^u \frac{\partial^2 v}{\partial x^2}(s,x(s)) d \langle x \rangle_s^{\tau}
 + \int_0^u \varrho(s,x(s))d x(s)\right]=
\end{equation*}
\begin{equation} \nonumber
  \lim_{u \to T} \left[ v(0, x(0))+ \int_0^u r[v(s,x(s))-\varrho(s,x(s))x(s)]ds +\int_0^u \varrho(s,x(s))d x(s) \right] =
\end{equation}
\begin{equation} \nonumber
 v(0, x(0))+ \int_0^T r \left[ v(s,x(s))-\varrho(s,x(s))x(s) \right] ds +\lim_{u \to T} \int_0^u \varrho(s,x(s))~dx(s) =
\end{equation}
\begin{equation} \nonumber
v(0, x(0)) + \int_0^T r \left[ v(s,x(s))-\varrho(s,x(s))x(s) \right] ds + \int_0^T \varrho(s,x(s))d x(s).
\end{equation}
The analysis in (\ref{eq:delta-hedging}) implies that the NP-portfolio
\begin{equation} \label{deltaHedgingPortfolio}
\Phi_t= \left(v(t,x(t))-\varrho(t,x(t))x(t), \varrho(t,x(t)) \right)
\end{equation}
replicates the payoff $h$
at maturity time $T$.
\end{proof}

\begin{corollary}  \label{deltaHedgingAdmissible}
The delta hedging portfolio  given by (\ref{deltaHedgingPortfolio}) is NP-admissible and V-continuous relative to the uniform topology.
\end{corollary}
\begin{proof}
The self-financing and predictable properties follow from the
definition and constructions in Theorem \ref{Non-stoch-BS} by
noticing that $x(t) = x(t-)$. The portfolio is admissible, with
$A=0$, by the known property $v(t,x(t)) \geq 0$.
V-continuity follows from
$V_\Phi(T, x) = h(x(T))$ and the fact that $h$ is continuous.
\end{proof}

\subsubsection{Arbitrage in $\mathcal{J}_{\tau}^{\sigma}$}

We analyze next the problem of arbitrage in a market where possible
trajectories are in $\mathcal{J}_{\tau}^{\sigma}$. We will make use
Theorem \ref{main-no-arbitrage} applied to the Black and Scholes model
$(Z, A^Z)$; in particular, $A^Z= \mathcal{A}^Z_{BS}$, where $\mathcal{A}^Z_{BS}$ denotes
the admissible portfolios in the Black-Scholes stochastic market. Corollary \ref{isomophic-v-cont} will be used to show that a
large class of portfolios belong to $[\mathcal{A}^Z_{BS}]$; towards this end, we incorporate portfolio
strategies that depend on past values of the trajectory and not just
on the spot value (\cite{valkeila-2}).

\begin{definition}
A hindsight factor $g$ over some class of trajectories $\mathcal{J}$ is a mapping
$g:[0,T] \times \mathcal{J} \to \mathbb{R}$ satisfying:
\begin{itemize}
\item [i)] $g(t, \eta)=g(t, \tilde{\eta})$ whenever $\eta(s)=\tilde{\eta}(s)$ for all $0 \leq s \leq t$.
\item [ii)] $g( \cdot, \eta)$ is of bounded variation and continuous for every $\eta \in \mathcal{J}$.
\item [iii)] There is a constant $K$ such that for every continuous function $f$.
\begin{equation} \nonumber
\left| \int_0^t f(s)dg(s, \eta)-\int_0^t f(s)dg(s, \tilde{\eta})\right| \leq
K \max_{0 \leq r \leq t} f(r) \left\| \eta- \tilde{\eta}\right\|_{\infty}
\end{equation}

\end{itemize}
\end{definition}

Another definition of \cite{valkeila-2} are the smooth strategies introduced next.

\begin{definition}  \label{smoothStrategies}
A portfolio strategy $\Phi=(\psi_t, \phi_t)_{0 \leq t \leq T}$ over the
class of trajectories $\mathcal{J}$ is called smooth if:
\begin{itemize}
\item [i)] The number of assets held at time $t$, $\phi_t$, has the form
\begin{equation} \label{eq:smooth-cond}
\phi_t(x) = \phi(t, x) =G(t,x_t,g_1(t,x), \ldots, g_m(t,x))
\end{equation}
for all  $t \in [0,T]$ and for all $x \in \mathcal{J}$
where $G \in C^1([0,T] \times \mathbb{R} \times \mathbb{R}^m)$ and the $g_i$'s
are hindsight factors
\item [ii)] There exists $A>0$ such that $V_{\Phi}(t,x) \geq -A$
$\forall t \in [0,T]$ and $\forall x \in \mathcal{J}$.
\end{itemize}
\end{definition}
Given the notation and assumptions from Definition \ref{smoothStrategies},
an application of It\^{o}-F\"ollmer
formula (\ref{Ito-Follmer}) proves that the integrals
$\int_0^t \phi(s, x) dx(s)$ exist for all $t \in [0,T]$ if
$\Phi=(\psi_t, \phi_t)$ is smooth. Propositions \ref{bond-admiss}
and \ref{NP-bond-admiss} (stated and proven in Appendix \ref{technicalResults})
relate the smoothness condition in (\ref{eq:smooth-cond})
with the admissibility conditions of both, stochastic and NP portfolios
respectively.

\begin{proposition}\label{smooth-continuity}
If $\Phi$ is a smooth portfolio strategy over $\mathcal{J}_{\tau}^{\sigma}(x_0)$
and $d$ is the uniform distance then $\Phi$ is V-continuous.
\end{proposition}
\noindent
Proposition \ref{smooth-continuity} follows immediately from Lemma
4.5 in \cite{valkeila-2}.

\vspace{.1in}
\noindent
The following result is well known (\cite{freedman}), we present a proof for
completeness.
\begin{lemma}\label{trajectories-tube}
Let $y$ be a continuous function $y:[0,T]\rightarrow \mathbb{R}$
with $y(0)=0$. If $W$ is a Brownian motion defined on a probability
space $(\Omega, \mathcal{F}, P)$ then for all $\epsilon>0$,
\begin{equation}\label{posit}
P\left(\omega: \sup_{s \in
[0,T]}\left|W_s(\omega)-y(s)\right|<\epsilon \right)>0.
\end{equation}
\end{lemma}
\begin{proof}
Function $y$ is continuous on $[0,T]$, therefore is uniformly continuous so for
all $\epsilon>0$, there exists $\delta>0$ such that
\[
|t_2-t_1|< \delta \Longrightarrow |y(t_2)-y(t_1)|< \epsilon/3
\]
Let $M$ be an integer, $M>T/\delta$, and define points $s_i=iT/M$, for $i=0,\ldots,M$.
By definition $|s_{i+1}-s_i|< \delta$, so $|y(s_{i+1})-y(s_{i})|< \epsilon/3$.\\
Define for all $ 1\leq i \leq M$
\[
A_i=\left\{\omega:\sup_{s_{i-1}\leq t \leq s_i} \left|W_t-W_{s_{i-1}} \right| <\epsilon/2 \right\}
\]
\[
B_i=\left\{\omega: \left|(W_{s_i}-W_{s_{i-1}})-(y(s_i)-y(s_{i-1})) \right| < \epsilon/ 6 M\right\}
\]
and $\Omega_i=A_i B_i$. It is immediate, from results in \cite{He},
that $P(\Omega_i)>0$. On the other hand, it is obvious that events
$\Omega_i$ and $\Omega_j$ are independent for $i \neq j$ since
increments of Brownian motions on disjoint intervals are
independent. So
\begin{equation}\label{positive}
P\left(\prod_{i=1}^M \Omega_i\right)=\prod_{i=1}^M P(\Omega_i)>0
\end{equation}
We will prove now that
\begin{equation}\label{inclusion}
\prod_{i=1}^M \Omega_i \subset \left\{ \omega: \sup_{s \in [0,T]}\left|W_s(\omega)-y(s)\right|<\epsilon  \right\}
\end{equation}
 Let $\omega \in \prod_{i=1}^M \Omega_i$, then $\omega \in \prod_{i=1}^M B_i$ so for all $k=1,\ldots, M-1$,
applying triangular inequality:
\begin{equation}
\left|W_{s_k}(\omega)-y(s_k)\right| \leq \sum_{i=1}^k \left|(W_{s_i}(\omega)-W_{s_{i-1}}(\omega))-(y(s_i)-y(s_{i-1})) \right|<
k \epsilon/6M \leq \epsilon/6
\end{equation}
Also $\omega \in A_{k+1}$, so $\left|W_t(\omega)-W_{s_{k}}(\omega) \right| <\epsilon/2$ for all $t \in [s_k, s_{k+1}]$.
We also know that $\left|y(s_k)-y(t)\right| \leq \epsilon/3$ for all $t \in [s_k, s_{k+1}]$. Using again triangular
inequality:
\begin{equation}\label{inequality}
\left|W_t(\omega)-y(t)\right| \leq \left|W_t(\omega)-W_{s_{k}}(\omega) \right| + \left|W_{s_k}(\omega)-y(s_k)\right| +
    \left|y(s_k)-y(t)\right| < \epsilon/2 + \epsilon/6 + \epsilon/3=\epsilon
\end{equation}
is valid for all $t \in [s_k, s_{k+1}]$ for all $k$ so (\ref{inequality}) is valid for all $t \in [0,T]$ and (\ref{inclusion})
is true. From (\ref{inclusion}) and (\ref{positive}) we obtain (\ref{posit}) and the Lemma is proved
\end{proof}

\vspace{.1in}
A main consequence of Theorem \ref{main-no-arbitrage} and the previous definitions
and results is the following Theorem.

\begin{theorem}\label{no-arbitrage-BS}
Let $(Z, \mathcal{A}_{BS}^Z)$ be the Black-Scholes stochastic market defined by
\[
Z_t=x_0 e^{\left(\mu-\sigma^2/2 \right)t + \sigma W_t},
\]
where $\mu$ and $\sigma>0$ are constant real numbers, $W$ is a Brownian Motion, and
$\mathcal{A}_{BS}^Z$ is the class of all admissible strategies for $Z$. Consider the
class of trajectories $\mathcal{J}_{\tau}^{\sigma}$ with the uniform topology. We have:
\begin{itemize}
\item [i)] The NP market  $(\mathcal{J}_{\tau}^{\sigma}, [\mathcal{A}_{BS}^Z])$ is NP arbitrage-free.
\item [ii)] $[\mathcal{A}_{BS}^Z]$ contains: \\
a) the smooth strategies such that the hindsight factors $g_i$ satisfy that $g_i(t,X)$ are
$(\mathcal{F}_{t-})$-measurable,\\
b) delta hedging strategies.
 \end{itemize}
\end{theorem}
\begin{proof}
i) By the definition of $Z$ and $\mathcal{J}_{\tau}^{\sigma}$
clearly condition $C_0$ in Theorem \ref{main-no-arbitrage} is
satisfied. Also, condition $C_1$ from Theorem \ref{main-no-arbitrage} follows from Lemma
\ref{trajectories-tube}. As $(Z, \mathcal{A}_{BS}^Z)$ is
arbitrage-free (see for example \cite{delbaen}) then the NP market
$(\mathcal{J}_{\tau}^{\sigma}, [\mathcal{A}_{BS}^Z])$ is NP
arbitrage-free according to Theorem \ref{main-no-arbitrage}.

\noindent ii) Let $\Phi$ be a smooth strategy over
$\mathcal{J}_{\tau}^{\sigma}$ As the trajectories in
$\mathcal{J}_{\tau}^{\sigma}$ are continuous, condition (\ref{NP-pred-g})
in Proposition \ref{NP-bond-admiss} holds, therefore
$\Phi$ is NP-admissible; $\Phi$ is also V-continuous as consequence of
Proposition \ref{smooth-continuity}.
Define a.s. $\Phi^z$ as
$\Phi^z(t,\omega)=\Phi(t,Z(\omega))$; Proposition \ref{bond-admiss}
shows that the stochastic portfolio $\Phi^z$ is predictable, LCRL
and self-financing. The admissibility of $\Phi^z$ results from ii)
in Definition \ref{smoothStrategies}, hence $\Phi^z \in
\mathcal{A}_{BS}^Z$.  As $\Phi$ and $\Phi^z$ are isomorphic
and $\Phi$ is V-continuous, Corollary \ref{isomophic-v-cont} applies
so $\Phi$ is connected to $\Phi^z$ and $\Phi \in [\mathcal{A}_{BS}^Z]$.
For the hedging
strategies the same arguments apply and the V-continuity and
admissibility follow by an application of Corollary
\ref{deltaHedgingAdmissible}.
\end{proof}

\begin{remark} In the framework of Theorem \ref{no-arbitrage-BS},
where trajectories in $\mathcal{J}$ are continuous it is not
difficult to see that $\tilde{g}(t,x)=\min_{0 \leq s \leq t}x(s)$, as
well as the maximum and the average, are
hindsight factors over $\mathcal{J}$ (see \cite{valkeila-2}),
moreover $\tilde{g}(t,X)$ is
a $(\mathcal{F}_{t-})$-measurable random variable.
\end{remark}

\begin{remark}\label{simple-piece-wise}
It can be proved that $[\mathcal{A}_{BS}^Z]$ also contains simple
(piece-wise constant) portfolio strategies
satisfying
\[
\phi_t=\sum_{l=1}^L 1_{(s_{l-1},s_l]}(t)G(t,x(s_{l-1}))
\]
where $0=s_0 < s_1 < \cdots < s_L=T$, the $s_i$ are deterministic
and $G$ is $C^{1}$. This is consequence of Remark 4.6 of \cite{valkeila-2}.
\end{remark}

Theorem \ref{no-arbitrage-BS} is the analogous in our framework of
the known absence of arbitrage in the Black-Scholes model, a
property that in fact we use in the above proof.

In a classical stochastic framework, the absence of arbitrage is
equivalent to the existence of at least one risk neutral probability
measure, the next example shows a possible trajectory class which
has no obvious probabilistic counterpart.

\begin{example} \label{rational-end-points}
Define the class
\[
\mathcal{J}_{\tau, \mathbb{Q}}^{\sigma}=\left\{ x \in
\mathcal{J}_{\tau}^{\sigma}: x(T) \in \mathbb{Q}\right\}
\]
where $\mathbb{Q}$ is the set of rational numbers.\\
Consider $[\mathcal{A}_{BS}^Z]$ as defined in Theorem \ref{no-arbitrage-BS}.
Let  $\mathcal{A}^{V} \subset [\mathcal{A}_{BS}^Z]$ be the class of all
$V$-continuous portfolios in $[\mathcal{A}_{BS}^Z]$. Item $ii)$ in Theorem
\ref{no-arbitrage-BS}, $\mathcal{A}^{V}$ is a large class of portfolios
which also satisfies that the market $(\mathcal{J}_{\tau}^{\sigma}, \mathcal{A}^{V})$
is NP-arbitrage free. Let $\mathcal{A}^{V}_{\mathcal{J}_{\tau, \mathbb{Q}}^{\sigma}}$
be the restriction of portfolio strategies in $\mathcal{A}^{V}$ to
the subclass of trajectories $\mathcal{J}_{\tau, \mathbb{Q}}^{\sigma}$.
As $\mathcal{J}_{\tau, \mathbb{Q}}^{\sigma}$ is dense on $\mathcal{J}_{\tau}^{\sigma}$, applying
Proposition \ref{subclasses} we conclude that the market $(\mathcal{J}_{\tau,
\mathbb{Q}}^{\sigma},\mathcal{A}^{V}_{\mathcal{J}_{\tau, \mathbb{Q}}^{\sigma}})$ is NP-arbitrage-free.
\end{example}

The absence of arbitrage for model in Example
\ref{rational-end-points} and replicating portfolio in Theorem
\ref{Non-stoch-BS} imply that it is possible to price derivatives
using the Black-Scholes formula also for this model, even if there
is no obvious intuitive measure over the possible set of
trajectories. In fact, the set $\mathcal{J}_{\tau,
\mathbb{Q}}^{\sigma}$ has null probability under the Black-Scholes
model, therefore if a measure is defined over this set, it will not
be absolutely continuous with respect to the Wiener measure. Hence,
it is not clear  how to price derivatives under a stochastic model
following a risk neutral approach,
if the trajectories of the asset price process belong to
$\mathcal{J}_{\tau, \mathbb{Q}}^{\sigma}$.

\subsection{Non Probabilistic Geometric Poisson Model}  \label{hedgingAndNoArbitrageInClassesWithJumps}
This section studies hedging and arbitrage in specific examples of trajectory classes with jumps.
Denote by $\mathcal{N}([0, T])$ the collection of all
functions $n(t)$ such that there exists a non negative integer $m$
and positive numbers $0 < s_1 < \ldots < s_m < T$ such that
$n(t)=\sum_{s_i \leq t} 1_{[0,t]}(s_i)$. The function $n(t)$ is
considered as identically zero on $[0, T]$ whenever $m=0$.

The following class of real valued functions will be another example
of possible trajectories for the asset price.
\begin{itemize}

\item Given constants $\mu, a \in \mathbb{R}$ and $x_0 > 0$, let $ \mathcal{J}^{a,\mu}(x_0)$ to be the class of all functions
$x$ for which exists $n(t) \in \mathcal{N}([0, T])$ such that:
\begin{equation} \label{eq:poisson}
x(t)=x_{0}e^{\mu t}(1+a)^{n(t)}.
\end{equation}
\end{itemize}
The function $n(t)$ counts the number of jumps present in the path
$x$ until, and including, time $t$. Note also that the definition of
$ \mathcal{J}^{a, \mu}(x_0)$ does not depend on the particular
subdivision $\mathcal{T}$ used elsewhere in the paper.

The natural probabilistic counterpart for this model is the Geometric Poisson model
\begin{equation}  \nonumber 
Z_t = x_0 e^{\mu t}~(1 + a)^{N^p_t},
\end{equation}
where $ N^P = (N^P_t)$ is a Poisson process on a filtered
probability space $(\Omega, \mathcal{F}, \mathcal{F}_t, P)$. Notice
that $P(Z(w) \in \mathcal{J}^{a, \mu}(x_0)) =1$. Even if this
stochastic model has limited practical use in finance, it has
theoretical importance because, together with the Black-Scholes
model, they are the only exponential L\'evy models leading to
complete markets, see \cite{cont}.

\begin{remark}
The class $ \mathcal{J}^{a,\mu}(x_0)$ includes trajectories of processes different than
the Geometric Poisson model, in fact if $N$ is a renewal process, trajectories of the process
Z defined as $Z_t= x_0 e^{\mu t}~(1 + a)^{N_t}$ are also in $\mathcal{J}^{a, \mu}(x_0)$.
\end{remark}
A replicating portfolio for trajectories in $\mathcal{J}^{a,\mu}$
corresponds to the probabilistic-free version of the hedging
strategy associated to the Geometric Poisson model, see \cite{carr}.

 Suppose we have an European type derivative with payoff $h(x(T))$. For simplicity
we consider interest rate $r=0$. We are looking for a NP-admissible
portfolio strategy that perfectly replicates the payoff $h(x(T))$. The next Theorem provides the answer to this NP hedging question.
\begin{theorem} \label{only-jumps}
Let $\mathcal{J}^{\ast}$ be a class of possible trajectories for the asset price,
$\mathcal{J}^{\ast} \subset \mathcal{J}^{a,\mu}(x_0)$. Consider that $a\mu<0$ and let
$\lambda=-\mu/a$. Define $\tilde{F}(s,t)$ by:
\begin{equation}\nonumber 
\tilde{F}(t, s)=e^{-\lambda(T-t)}\sum_{k=0}^\infty \frac{h\left(s e^{\mu (T-t)}(1+a)^k \right)(T-t)^k}{k!}.
\end{equation}
Then, the portfolio $\Phi_t=(\psi_t,\phi_t)$ where
\begin{equation}\nonumber 
\phi_t=\frac{\tilde{F}(t, (a+1)x(t^-))-\tilde{F}(t, x(t^-))}{a~x(t^-)},
\end{equation}
and $\psi_t=\tilde{F}(t, x(t^-))- \phi_t x(t^-)$,
whose initial value is $\tilde{F}(0, x_0)$, replicates the Lipschitz payoff $h(x(T))$ at time $T$ for every
$x \in \mathcal{J}^{\ast}$.
\end{theorem}

We will not provide a proof of Theorem \ref{only-jumps} as it can be easily extracted from \cite{carr} even though that reference obtains a probabilistic result considering
$n(t)$ (as appears in equation  (\ref{eq:poisson})) to be a Poisson process $N^P_t$. The proof in  \cite{carr} can be translated
to our non probabilistic model in a straightforward way. It is important to remark that such a proof would use only ordinary calculus.

Theorem \ref{only-jumps}  can be generalized to the case where  $a$ and $\mu$ are considered no longer
as constants but known deterministic functions $a(t)$ and $\mu(t)$ such that $\mu(t) a(t) < 0$
for all $t$. A proof of this result is contained in \cite{alvarez}.

\subsubsection{Arbitrage in $\mathcal{J}_{}^{a, \mu}$}

Next we concentrate on establishing the absence of NP arbitrage for
the class of trajectories $\mathcal{J}^{a, \mu}$ to this end we will
apply Theorem \ref{main-no-arbitrage}. It remains to select an
appropriate metric $d$. Instead of using the uniform distance, for
models with jumps we will use the Skorohod's distance $d_s$; for a
definition of this distance and its associated topology we refer the
reader to \cite{billingsley}.

The next proposition gives sufficient conditions for the V-continuity
of a portfolio over $\mathcal{J}^{a, \mu}$ with respect to the Skorohod's
metric.

\begin{proposition}  \label{sContinuousStrategies}
A portfolio $\Phi_t=(\psi_t, \phi_t)$ on $\mathcal{J}^{a, \mu}$ for which
the amount invested in the stock $\phi_t=\phi(t,x_{t-})$ is such that
$\phi \in C([0,T]\times\mathbb{R})$ is V-continuous relative to the Skorohod's
topology.
\end{proposition}

\begin{proof}
Let $x \in \mathcal{J}^{a,\mu}$ and let $\left \{ x^{(n)} \right \}_{n=0,1...}$ with
$x^{(n)} \in \mathcal{J}^{a,\mu}$ be a sequence that converges to $x$ in the
Skorohod's distance.
Suppose that $x$ has $m$ jumps located at $0< \tau_1 < \ldots <\tau_m <T$. Then for every
$\epsilon>0$ there exists an integer $K>0$ such that for $n>K$, $x^{(n)}$ has exactly
$m$ jumps located at $0< \tau_1^{(n)} < \ldots <\tau_m^{(n)} <T$ and satisfying
that $|\tau_i^{(n)}-\tau_i|<\epsilon$ for $i=1,\ldots,m$. For convenience denote
$\tau_0=\tau_0^{(n)}=0$ and $\tau_{m+1}=\tau_{m+1}^{(n)}=T$.

Next we evaluate $V_{\Phi}(T,\cdot)$, the value of portfolio $\Phi$ at maturity time $T$
for both trajectories $x$ and $x^{(n)}$. Because we are restricting to the case of  interest rate $r=0$, we have
$\forall x \in \mathcal{J}^{a,\mu}$:

\begin{equation}\label{final-value-x}
V_{\Phi}(T,x) = V_{\Phi}(0,x) + \int_0^T \phi^s d x_s = V_0+ \int_0^T \phi(s,x_{s-})dx_s
\end{equation}

Using the particular form of $x$, the integral on (\ref{final-value-x}) can be computed as:

\begin{eqnarray}\label{int-x}
&& \int_0^T \phi(s,x_{s-})dx_s =
\sum_{i=0}^{m} \int_{\tau_i}^{\tau_{i+1}} \phi\left(s,x_0 e^{\mu s}(1+a)^i\right) \mu x_0 e^{\mu s}(1+a)^i ds \\
\nonumber
 && + \sum_{i=1}^{m} \left[ \phi(\tau_i,x_0 e^{\mu \tau_i}(1+a)^i)-\phi(\tau_i,x_0 e^{\mu \tau_i}(1+a)^{i-1}) \right]
 a x_0 e^{\mu \tau_i}(1+a)^{i-1}
\end{eqnarray}

A similar expression applies for $x^{(n)}$ for all $n$:

\begin{equation}\label{int-xn}
  \int_0^T \phi(s,x_{s-}^{(n)})dx_s^{(n)} =
\sum_{i=0}^{m} \int_{\tau_i^{(n)}}^{\tau_{i+1}^{(n)}} \phi\left(s,x_0 e^{\mu s}(1+a)^i\right) \mu x_0 e^{\mu s}(1+a)^i ds \\
+
\end{equation}
\begin{equation} \nonumber
\sum_{i=1}^{m} \left[ \phi(\tau_i^{(n)},x_0 e^{\mu \tau_i^{(n)}}(1+a)^i)-
                            \phi(\tau_i^{(n)},x_0 e^{\mu \tau_i^{(n)}}(1+a)^{i-1}) \right]
 a x_0 e^{\mu \tau_i^{(n)}}(1+a)^{i-1}
\end{equation}

As $\tau_i^{(n)} \to \tau_i$ as $n \to \infty$ integrals and summands in (\ref{int-xn})  converge
to analogous elements in (\ref{int-x}), thus $V_{\Phi}(T,x^{(n)}) \to V_{\Phi}(T,x)$, which
proves that $\Phi$ is a V-continuous portfolio.
\end{proof}

The next Theorem shows that our general Theorem
\ref{main-no-arbitrage} is also useful to establish the
absence of NP-arbitrage in models with jumps.

\begin{theorem}   \label{no-arbitrage-jumps}
Let  $(Z, \mathcal{A}^Z_P)$ be the stochastic market defined by the
geometric Poisson stochastic process introduced before:
\begin{equation}  \label{poissonProcess}
Z_t = x_0 e^{\mu~t}~(1 + a)^{N^p_t},
\end{equation}
and $\mathcal{A}^Z_P$ is the class of admissible strategies for $Z$.
Consider the class of trajectories $\mathcal{J}^{a,\mu}$ endowed with the
Skorohod's topology. We have:
\begin{itemize}
\item [i)] The NP market  $(\mathcal{J}_{\tau}^{a, \mu}, [\mathcal{A}_{P}^Z])$ is NP arbitrage-free.
\item [ii)] $[\mathcal{A}_{P}^Z]$ contains the portfolio strategies from Proposition \ref{sContinuousStrategies} which furthermore satisfy that there exist $A>0$ such that $V_{\Phi}(t,x)>-A$
$\forall t \in [0,T]$, $\forall x \in \mathcal{J}_{\tau}^{a, \mu}$. $[\mathcal{A}_{P}^Z]$ also contains the portfolio strategies defined in Theorem \ref{only-jumps}.
\end{itemize}
\end{theorem}
\begin{proof}
$i)$ The proof is analogous to the proof of Theorem
\ref{no-arbitrage-BS}. As indicated, $P(Z(\omega) \in
\mathcal{J}^{a,\mu}(x_0))=1$, so condition $C_0$ from Theorem
\ref{main-no-arbitrage} holds.
In order to verify condition $C_1$ from Theorem
\ref{main-no-arbitrage}, we argue directly (another
possibility would be to extract the result from the proof of
the more general Lemma \ref{smallballlevy}). Consider
$x \in \mathcal{J}^{a,\mu}(x_0)$ and
suppose that $x$ has $m$ jumps at times $0<s_1 \cdots <s_m<T$. Then,
for all $\epsilon>0$ there exists $\delta>0$ such that if $x' \in
\mathcal{J}^{a,\mu}(x_0)$ has exactly $m$ jumps $0<s'_1 \cdots
<s'_m<T$ with $|s_i-s'_i|< \delta$ then $d_s(x,x')<\epsilon$. We
know that the time between two consecutive jumps of a Poisson
process has exponential distribution (which is absolutely continuous
with respect to the Lebesgue measure on $\mathbb{R}^+$) therefore,
jumps occur in a given interval with positive probability. From
previous analysis and the property of independence of the time
between jumps, we conclude that the set of trajectories $x'$ of the
Geometric Poisson model (\ref{poissonProcess}) having jumps in a
$\delta$-neighborhood of the jumps of any $x \in
\mathcal{J}^{a,\mu}(x_0)$ has positive probability, thus Condition
$C_1$ is verified and so the market $(\mathcal{J}^{a,\mu}(x_0),
[\mathcal{A}_P^Z])$ is NP-arbitrage-free.

\noindent $ii)$ Consider $\Phi$ from Proposition \ref{sContinuousStrategies} satisfying the lower bound assumption  then,
the same arguments used in the proof of $ii)$ of Theorem
\ref{no-arbitrage-BS} apply in this case, namely, the use of Propositions
\ref{bond-admiss}, \ref{NP-bond-admiss} and \ref{sContinuousStrategies},
as well as Corollary  \ref{isomophic-v-cont} prove that
$\Phi \in [\mathcal{A}^Z_P]$. Similar arguments show that
$\Phi \in [\mathcal{A}^Z_P]$ whenever $\Phi$
is one of the hedging strategies introduced in Theorem
\ref{only-jumps}.
\end{proof}

\noindent One important question at this point is whether or not simple portfolio strategies
are V-continuous for the Geometric Poisson model. Next proposition addresses that question.

\begin{proposition}\label{counterex-simple}
In general, simple portfolios strategies are not V-continuous relative to the Skorohod topology in $\mathcal{J}^{a, \mu}(x_0)$.
\end{proposition}
\begin{proof}
We provide an example of a simple strategy that is not V-continuous.
Consider $T=1$ and let $\Phi$ the NP-portfolio with initial value
$x_0$ defined as:
\begin{itemize}
\item $\phi(t,x)=1$, $\psi(t,x)=0$, for all
$x \in \mathcal{J}^{a, \mu}(x_0)$ if $0 \leq t \leq 1/2$
\item $\phi(t,x)=0$, $\psi(t,x)=x_{\frac{1}{2}}$ for all $x \in \mathcal{J}^{a, \mu}(x_0)$ if $1/2<t \leq 1$
\end{itemize}
We can easily check that $\Phi$ is NP-admissible according to Definition
\ref{def:NP-conditions}.

Let $y \in \mathcal{J}^{a, \mu}(x_0)$ be the function
$y_t=x_{0}e^{\mu t}(1+a)^{1_{[0,t]}(1/2)}$ and let $(y^{(n)})_{n=1,\ldots}$ be the
sequence of functions defined by $y_t^{(n)}=x_{0}e^{\mu t}(1+a)^{1_{[0,t]}(1/2+1/n)}$.
Clearly $y^{(n)} \to y$ in the Skorohod topology on $D[0,1]$.

From the definition of $\Phi$ we have $V_{\Phi}(1,x)=x_{\frac{1}{2}}$ for all
$x \in \mathcal{J}^{a, \mu}(x_0)$, in particular $V_{\Phi}(1,y)=x_{0}e^{\mu/2}(1+a)$
and $V_{\Phi}(1,y^{(n)})=x_{0}e^{\mu/2}$.

As $y^{(n)} \to y$ but $V_{\Phi}(1,y^{(n)}) \nrightarrow V_{\Phi}(1,y)$ we conclude that
$\Phi$ is not a V-continuous portfolio in $\mathcal{J}^{a, \mu}(x_0)$.
\end{proof}

\noindent It can also be shown that not necessarily a portfolio must be V-continuous in order to
belong to $[\mathcal{A}^Z_P]$. An example of such portfolio is given in the next proposition,
proving that $[\mathcal{A}^Z_P]$ contains more portfolios than those explicitly showed in Theorem
\ref{no-arbitrage-jumps}.

\begin{proposition}
If $a<0$ the portfolio $\Phi$ in Proposition \ref{counterex-simple} belongs to $[\mathcal{A}^Z_P]$
\end{proposition}

\begin{proof}
Let $\Phi$ the NP-portfolio defined in Proposition \ref{counterex-simple}, and consider the
isomorphic portfolio $\Phi^Z$ over the price process $Z$ in (\ref{poissonProcess})
defined a.s. by $\Phi^Z(t, \omega)=\Phi(t,Z(\omega))$. The portfolio
strategy $\Phi^Z$ is a simple strategy, therefore
$\Phi^Z$ is admissible so $\Phi^Z \in \mathcal{A}^Z_P$. Let us show that even if
$\Phi$ is not V-continuous on $\mathcal{J}^{a, \mu}(x_0)$, $\Phi \in [\mathcal{A}^Z_P]$.

If an arbitrary trajectory $x \in \mathcal{J}^{a, \mu}(x_0)$ is continuous at $t=1/2$, then it
is always possible to choose $\delta>0$ small enough, such that for all $x'$ satisfying
$d_s(x',x)<\delta$, it holds $x'(1/2)=x(1/2)$ which in turns implies that
$V_{\Phi}(1,x')=V_{Phi}(1,x)$. If an arbitrary $x$ is discontinuous at $t=1/2$, then
$\Phi$ is not V-continuous at $x$, in fact that was the statement of
Proposition \ref{counterex-simple}. Nevertheless, for such trajectories $x$ we can always choose
$\delta>0$ small enough such that for all $x'$ satisfying
$d_s(x',x)<\delta$, we have one of the two following possibilities:\\
1) $x'(1/2)=x(1/2)$.\\
2) $x'(1/2)=x(1/2)(1+a)^{-1}$.\\
Case 1) corresponds to those trajectories $x'$ that jump in the interval $(1/2-\delta, 1/2]$
while case 2) corresponds to those trajectories that jump in $(1/2, 1/2+ \delta)$.
For trajectories $Z(\omega)$ in case 1) we have $V_{\Phi}(1,x')=V_{\Phi}(1,x)$.
For those in case 2) it holds  $V_{\Phi}(1,x')=x(1/2)(1+a)^{-1}>x(1/2)=V_{\Phi}(T,x)=x(1/2)$
if $a<0$.
What we have shown with previous analysis is that for an arbitrary trajectory $x \in \mathcal{J}^{a, \mu}(x_0)$,
whether or not $x$ is continuous at $t=1/2$, it is always possible to find $\delta$ small enough
such that if $d_s(x',x)<\delta$ then $V_{\Phi}(1, x')\geq V_{\Phi}(1,x)$, which
implies that application $V_{\Phi}(1, \cdot): \mathcal{J}^{a, \mu}(x_0) \rightarrow \mathbb{R}$
is lower semicontinuous with respect to the Skorohod topology. Applying
Proposition \ref{sufficientConditionForNPPortfolio}, as $\Phi^Z \in \mathcal{A}^Z_P$, then
$\Phi$ is connected to $\Phi^Z$ hence $\Phi \in [\mathcal{A}^Z_P]$.
\end{proof}

\begin{remark}
It is expected that the present class $[\mathcal{A}^Z_P]$ could be considerably enlarged, and in particular
simple strategies would belong to this enlarged class, once the notion of stopping times is incorporated in our non
probabilistic approach. This line of research represents work in progress \cite{alvarez}.
\end{remark}

\subsection{Non Probabilistic Jump Diffusions}
Fix $\sigma > 0$ and $C$ a non empty set of real numbers such
that $\inf(C)>-1$.  Define $\mathcal{J}_{\tau}^{\sigma, C}(x_0)$ as
the class of real valued functions $x$ on $[0,T]$ such that there exits
$z \in \mathcal{Z}_{\mathcal{T}}([0, T])$, $n(t) \in \mathcal{N}([0,
T])$,  and real numbers $a_i \in C$, $i=1,2,\ldots,m$, verifying:
\begin{equation} \label{eq:diffpoisson}
x(t)=x_0 e^{\sigma z(t)}\prod_{i=1}^{n(t)} (1+a_i)
\end{equation}

\begin{remark} \label{generalizedClasses}
The class $\mathcal{J}_{\tau}^{\sigma, a}(x_0)$ combines the features of
classes given in Sections \ref{ssect:cont} and
\ref{hedgingAndNoArbitrageInClassesWithJumps}. Its probabilistic
counterpart is the  class of exponential jump-diffusion processes.
\end{remark}

In the stochastic framework, the markets where prices are driven by
jump-diffusion models are not complete in general, therefore hedging is not
always possible. On the other hand, we do know that these models
admit many risk neutral measures, indicating that they are arbitrage free.
In this section we will obtain the property of absence
of arbitrage in the analogous NP  framework given by trajectories belonging to
$\mathcal{J}_{\tau}^{\sigma, C}(x_0)$.

First we give a small ball property result
for a jump-diffusion model and the class of price trajectories
$\mathcal{J}_{\tau}^{\sigma, C}(x_0)$ defined in
(\ref{eq:diffpoisson}) and then we derive the NP arbitrage-free result
using  Theorem \ref{main-no-arbitrage} by
relying on the V-continuity property of certain class of portfolios
with respect to the Skorohod topology.

\noindent The following proposition  provides the required small ball
property.
\begin{proposition}\label{smallballlevy}
For any $x_0>0$ consider in the probability space $(\Omega,
\mathcal{F}, (\mathcal{F}_t)_{t \geq 0}, P)$  the  exponential jump
diffusion processes, starting at $x_0$ given by:
\begin{equation}\label{eq:expjumpdiff}
     Z_t=x_0 e^{(\mu-\frac{1}{2}\sigma^2)t+\sigma W_t}\prod_{i=1}^{N_t} (1+X_i)
\end{equation}
where $W$ is a Brownian Motion, $N$ is a homogeneous Poisson Process  with intensity
$\lambda >0$,
and the $X_i$ are independent random variables, also
independent of $W$ and $N$, with common probability distribution
$F_X$. Assume that $F_X$ verifies the condition:\\
\\
A1)For any $a \in C$ and for all $\epsilon >0$, $F_X(a+\epsilon)-F_X(a-\epsilon)>0$.\\
\\
\noindent Then the jump-diffusion process
given by (\ref{eq:expjumpdiff}) satisfies a small ball property on
$\mathcal{J}_{\tau}^{\sigma, C}(x_0)$ with respect to the Skorohod
metric.
\end{proposition}

\begin{proof}
Consider $x(\cdot)=x(0)e^{\sigma z(\cdot)}\prod_{i=1}^{n(\cdot)}(1+a_i) \in \mathcal{J}_{\tau}^{\sigma, C}(x_0)$, where
$n(\cdot) \in \mathcal{N}([0,T])$ has $m$ discontinuity points in $[0,T]$ denoted
 by $0 < s_1< \ldots< s_m <T$. Also denote $\tilde{n}(t)=\sum_{i=1}^{n(t)} \ln(1+a_i)$ and
 $\xi_t=\sum_{i=1}^{N(t)} \ln(1+X_i)$. \\
\\
Fix $\epsilon>0$ and consider $\delta >0$.
Define  $\Omega_1^{\delta}$ as the set of $w \in \Omega $ having jump times
 $0<T_1(w)<T_2(w)<\ldots < T_m(w)<T$ and $T_{m+1}>T$,  satisfying also that
$ |T_i(w)-s_i| \leq \frac{\delta}{3}$, for $i=1,2,\ldots,m$.
Note that the $T_i$'s are finite sum of continuous random variables, namely the times between
jumps, therefore we have that $P(\Omega_1)>0$.\\
\\
Take now  $\Omega_2^{\delta}$  as the set of $w \in \Omega$ such that:\\
$|\ln(1+X_i(w))-\ln(1+a_i)| < \frac{\delta}{3m}$, for $i=1,2, \dots,m$, which implies that
\begin{equation*}
   \left|\sum_{i=1}^m \ln(1+X_i(w)) -\sum_{i=1}^m \ln(1+a_i)\right|< \frac{\delta}{3}.
\end{equation*}
  Let $\lambda(t)$  be function from $[0,T]$ onto itself defined by the polygonal through the points  $(0,0), (s_1,T_1(w)), \ldots, (s_m, T_m(w)), (T,T)$. \\
Note that, by construction, for  $ w \in \Omega_1^{\delta} \bigcap \Omega_2^{\delta} $ we have:
\begin{equation} \nonumber 
  \sup_{t \in [0,T]} |\lambda(t)-t|< \frac{\delta}{3}
 \end{equation}
 \begin{equation}\nonumber 
    \sup_{t \in [0,T]} |\xi_{\lambda(t)}(w)-\tilde{n}(t)| < \frac{\delta}{3}.
\end{equation}
We should note that $P(\Omega_2^{\delta})>0$ for every $\delta>0$ as consequence of condition A1).\\
For $0 \leq t \leq T$ we have
\begin{eqnarray}\label{eq:triang-ineq}
&& \left|\left(\mu-\frac{1}{2}\sigma^2\right) \lambda(t)+\sigma W_{\lambda(t)}+ \xi_{\lambda(t)}-\sigma z(t) - \tilde{n}(t)\right|\\  &\leq& \sigma \left|W_{\lambda(t)}-z(t)  +\frac{1}{\sigma}\left(\mu-\frac{1}{2}\sigma^2\right)\lambda(t) \right| + | \xi_{\lambda(t)}- \tilde{n}(t)|  \nonumber \\
 &\leq& \sigma \left|W_{\lambda(t)}-z(\lambda(t))  +\frac{1}{\sigma}\left(\mu-\frac{1}{2}\sigma^2\right)\lambda(t) \right| + \sigma |z(\lambda(t))-z(t)|+| \xi_{\lambda(t)}- \tilde{n}(t)|. \nonumber
 \end{eqnarray}
Define $z'(t)=z(t)-\frac{1}{\sigma}\left(\mu-\frac{1}{2}\sigma^2\right)t$ and let $\Omega_3^{\delta}=\{ \omega \in \Omega: \sup_{t \in [0,T]}
|W_t(w)-z'(t)| < \frac{\delta}{3\sigma} \}$. By Lemma \ref{trajectories-tube} we have  $P(\Omega_3^{\delta})> 0$. Moreover by independence between $W$, $N$ and the $X_i$'s we have that:
\begin{equation*}
    P(\Omega_1^{\delta} \bigcap \Omega_2^{\delta} \bigcap \Omega_3^{\delta})=
    P( \Omega_1^{\delta})P(\Omega_2^{\delta})P(\Omega_3^{\delta})>0
\end{equation*}
Now, as $z(t)$ is uniformly continuous on $[0,T]$, there exists $\delta'>0$ such that  $|\lambda(t)-t|< \delta'$
implies  $\sup_{t \in [0,T]} |z(\lambda(t))-z(t)|<\frac{\epsilon}{3\sigma}$. Without loss of generality  take
$0< \delta< \min(\epsilon,\delta')$\\

\noindent According to (\ref{eq:triang-ineq}), in $\Omega_1^{\delta} \bigcap \Omega_2^{\delta} \bigcap \Omega_3^{\delta}$ we have, for $0 \leq t \leq T$:
\begin{eqnarray}
&& \left|\left(\mu-\frac{1}{2}\sigma^2\right) \lambda(t)+\sigma W_{\lambda(t)}+ \xi_{\lambda(t)}-\sigma z(t) - \tilde{n}(t)\right| \nonumber \\
 &\leq& \sigma \left|W_{\lambda(t)}-z'(\lambda(t))\right| + \sigma |z(\lambda(t))-z(t)|+| \xi_{\lambda(t)}- \tilde{n}(t)|. \nonumber \\
 &\leq& \frac{\delta}{3}+\frac{\epsilon}{3}+\frac{\delta}{3} < \epsilon \nonumber
\end{eqnarray}

As $Z_t=\exp \left( (\mu-\frac{1}{2}\sigma^2)t+\sigma W_{t}+ \xi_{t}\right)$ and
$x(t)=\exp \left(\sigma z_t + \tilde{n}(t) \right)$, the small ball property is obtained
by an argument of uniform continuity of the exponential function.
\end{proof}

Analogously to Proposition \ref{sContinuousStrategies} we have the
following result on V-continuity relative to the Skorohod's topology on
$\mathcal{J}_{\tau}^{\sigma,C}(x_0)$.

\begin{proposition} \label{sContinuousStrategiesII}
Let $\Phi_t=(\psi_t, \phi_t)$ be a portfolio strategy on $\mathcal{J}_{\tau}^{\sigma,C}(x_0)$ for which
the amount invested in the stock $\phi_t=\phi(t,x(t-))$ is such that $\phi \in C^{1,1}([0,T]\times R^+)$
and $\psi$ is defined through the self financing condition in (\ref{NP-bond-invest}).
Assume also that $\inf_{c \in C}|c|>h$ for some real number $h$.
Then $\Phi$ is V-continuous on $\mathcal{J}_{\tau}^{\sigma,C}(x_0)$ relative to the Skorohod's topology.
\end{proposition}
\begin{proof}
Let $\Phi_t=(\psi_t, \phi_t)$ be such that
$\phi_t=\phi(t,x(t-))$ with $\phi \in C^{1,1}([0,T]\times R^+)$.
Define the function $U_{\Phi}: \mathbb{R}^2 \rightarrow \mathbb{R}$ as:

\begin{equation}
U_{\Phi}(t,x)=\int_{x_0}^x \phi(t, \xi) d \xi
\end{equation}

and the functional $u_{\Phi}: \mathcal{J}_{\tau}^{\sigma,C}(x_0) \rightarrow \mathbb{R}$ as:

\begin{eqnarray}\label{functional-u}
u_{\Phi}(x)&=&U_{\Phi}(T,x(T))-U_{\Phi}(0,x(0))\\
           &&-\int_0^T \frac{\partial U_{\Phi}}{\partial t}(s,x(s-))ds
             -\frac{1}{2}\int_0^T \frac{\partial^2 U_{\Phi}}{\partial x^2}(s,x(s-))d \langle x \rangle_s^{\mathcal{T}} \nonumber \\
           &&-\sum_{s \leq T} \left[ U_{\Phi}(s,x(s))-U_{\Phi}(s,x(s-))
             -\frac{\partial U_{\Phi}}{\partial x}(s,x(s-))\Delta x(s)\right] \nonumber
\end{eqnarray}

From Ito-F\"ollmer formula

\begin{equation}
u_{\Phi}(x)=\int_0^T \frac{\partial U_{\Phi}}{\partial x}(s,x(s-)) dx(s)=\int_0^T \phi(s,x(s-)) dx(s)
\end{equation}
which implies that portfolio $\Phi$ is V-continuous on $\mathcal{J}_{\tau}^{\sigma,C}(x_0)$
(with respect to the Skorohod's topology) if
and only if the functional $u_{\Phi}$ is continuous on $\mathcal{J}_{\tau}^{\sigma,C}(x_0)$
with respect to the Skorohod's topology.

For all $x \in \mathcal{J}_{\tau}^{\sigma,C}(x_0)$, $d \langle x \rangle_s^{\mathcal{T}}= \sigma^2 x^2(s-)ds$,
therefore (\ref{functional-u}) transforms into:
\begin{equation}
u_{\Phi}(x)=U_{\Phi}(T,x(T))-U_{\Phi}(0,x(0))-I_{\Phi}(x)-S_{\Phi}(x)
\end{equation}
where
\begin{equation}\nonumber 
I_{\Phi}(x)=\int_0^T \frac{\partial U_{\Phi}}{\partial t}(s,x(s-))ds
             +\frac{1}{2}\int_0^T \frac{\partial^2 U_{\Phi}}{\partial x^2}(s,x(s-))\sigma^2 x^2(s-)ds
\end{equation}
and
\begin{equation} \nonumber
S_{\Phi}(x)=\sum_{s \leq T} \left[ U_{\Phi}(s,x(s))-U_{\Phi}(s,x(s-))
            +\frac{\partial U_{\Phi}}{\partial x}(s,x(s-))\Delta x(s)\right].
\end{equation}

Let $x^* \in \mathcal{J}_{\tau}^{\sigma,C}(x_0)$ and let $\{x^{(n)}\}_{n=0,1,...}$ be a sequence of
functions, $x^{(n)} \in \mathcal{J}_{\tau}^{\sigma,C}(x_0)$, such that $\{x^{(n)}\}$ converges
to $x^*$ in the Skorohod's topology.

From Lemma \ref{lemma:neighboors} and the
continuity of both $U_{\Phi}$ and $\frac{\partial U_{\Phi}}{\partial x}$,
it is immediate to see that $S_{\Phi}(x^{(n)}) \to S_{\Phi}(x^*)$.\\

Next we will prove that $I_{\Phi}(x^{(n)}) \to I_{\Phi}(x^*)$. Consider
\[
g(s,x)=\frac{\partial U_{\Phi}}{\partial t}(s,x)+\frac{\sigma^2~x^2}{2}\frac{\partial^2 U_{\Phi}}{\partial x^2}(s,x)
\]

Then

\begin{equation}\label{diff-I}
\left| I_{\Phi}(x^*)- I_{\Phi}(x^{(n)})\right|= \left|  \int_0^T g(s,x^*(s)-)ds - \int_0^T g(s,x^{(n)}(s)-)ds\right|.
\end{equation}

Let $(\lambda_n)$ be a sequence of non-decreasing functions on $[0,T]$ such that $\lambda_n(0)=0$, $\lambda_n(T)=T$ and $ \lambda_n(s) \to s$.
Making the variable change $s\rightarrow \lambda_n(s)$ on the second integral, expression (\ref{diff-I})
transforms into:

\begin{eqnarray}\nonumber 
\left| I_{\Phi}(x^*)- I_{\Phi}(x^{(n)})\right|
   &=& \left|  \int_0^T g(s,x^*(s-))ds - \int_0^T g\left(\lambda_n(s),x^{(n)}(\lambda_n(s)-)\right)d\lambda_n(s)\right|  \nonumber \\
&\leq& \left|  \int_0^T g(s,x^*(s-))ds - \int_0^T g(s,x^*(s-))d\lambda_n(s) \right| \nonumber \\
    &&+\left|  \int_0^T g(s,x^*(s-))d\lambda_n(s) -
               \int_0^T g\left(\lambda_n(s),x^{(n)}(\lambda_n(s)-)\right)d\lambda_n(s)\right| \nonumber\\
&\leq& \left|  \int_0^T g(s,x^*(s-))ds - \int_0^T g(s,x^*(s-))d\lambda_n(s) \right| \nonumber \\
 &&+  \int_0^T \left| g(s,x^*(s-))-g\left(\lambda_n(s),x^{(n)}(\lambda_n(s)-)\right)\right| d\lambda_n(s).\nonumber
\end{eqnarray}

As $\left(\lambda_n(s),x^{(n)}(\lambda_n(s)-)\right)$ converges to $(s,x^*(s-))$ uniformly in s and $g$ is continuous, hence
uniformly continuous on compact sets then:

\begin{equation}\label{eq:sum1}
\int_0^T \left| g(s,x^*(s-))-g\left(\lambda_n(s),x^{(n)}(\lambda_n(s)-)\right)\right| d\lambda_n(s) \to 0
\end{equation}

Notice also that

\begin{equation}\label{eq:sum2}
\int_0^T g(s,x^*(s-))d\lambda_n(s) \to  \int_0^T g(s,x_s-)ds
\end{equation}

Expression (\ref{eq:sum2}) is consequence of the weak convergence of $\lambda_n(s)$ to $s$
and the fact that $y(s)=g(s,x(s-))$ is bounded on $[0,T]$ with only a finite number of
discontinuities. See for example the Continuous Mapping Theorem  (page 87, \cite{durret}).\\

Both (\ref{eq:sum1}) and (\ref{eq:sum2}) imply that $I_{\Phi}(x^{(n)}) \to I_{\Phi}(x^*)$. As
already proved $S_{\Phi}(x^{(n)}) \to S_{\Phi}(x^*)$, therefore
$u_{\Phi}(x^{(n)}) \to u_{\Phi}(x^*)$, and the V-continuity of $\Phi$ is proved

\end{proof}

We have now all necessary ingredient to prove an arbitrage result for the present class of trajectories.

\begin{theorem}\label{noArbitrageHedginStrategies}
Let  $(Z, \mathcal{A}^Z_{JD})$ be the stochastic market defined by
the geometric jump diffusion process introduced in (\ref{eq:expjumpdiff})
and $\mathcal{A}^Z_{JD}$ is the class of admissible strategies for
$Z$.
Consider the class of trajectories $\mathcal{J}_{\mathcal{T}}^{\sigma, C}$
introduced in (\ref{eq:diffpoisson}) endowed with the Skorohod's topology.
Assume the random variables $X_i$ to be integrable and that their common probability distribution  $F_X$ satisfies the following conditions:\\
1) $supp(F_X) \subset C$ where $supp(F)$ stands for the support of the distribution function $F$\\
2) For any $a \in C$ and for all $\epsilon >0$, $F_X(a+\epsilon)-F_X(a-\epsilon)>0$\\
Then it holds:
\begin{itemize}
\item [i)]The NP market  $(\mathcal{J}_{\tau}^{\sigma, C}, [\mathcal{A}_{JD}^Z])$ is NP arbitrage-free\\
\item [ii)]If $C$ satisfies that $\inf(C)>-1$ and $\inf_{c \in C}|c|>h$ for some real
number $h >0$, then $[\mathcal{A}_{JD}^Z]$ contains the portfolios from Proposition \ref{sContinuousStrategiesII} which furthermore satisfy that there exist $A>0$ such that $V_{\Phi}(t,x)>-A$
$\forall t \in [0,T]$, $\forall x \in \mathcal{J}_{\tau}^{\sigma, C}$.
\end{itemize}
\end{theorem}

\begin{proof}
First note that $P(w \in \Omega : Z(w) \in \mathcal{J}_{\tau}^{\sigma, C})=1$
as consequence of 1) hence condition $C_0$ from Theorem
\ref{main-no-arbitrage} is fulfilled. From Proposition \ref{smallballlevy} we also have that
condition $C_1$ from Theorem
\ref{main-no-arbitrage} holds.
Therefore, in order to establish conclusion $i)$ we need to argue
that the stochastic market $(Z, \mathcal{A}^Z_{JD})$ is arbitrage
free. Our hypothesis allow the application of Proposition 9.9 from \cite{cont}, this result establishes the existence of a probability $\mathbb{Q}$ such that $e^{r t}~Z_t$
is a martingale, therefore the probabilistic market
$(Z,\mathcal{A}_{JD}^Z)$ is arbitrage free.

\noindent ii) Consider $\Phi$ satisfying the conditions listed in $ii)$ and define $\Phi^z$ as
$\Phi^z(t,\omega)=\Phi(t,Z(\omega))$
 Proposition \ref{bond-admiss}
shows that the stochastic portfolio $\Phi^z$ is predictable, LCRL
and self-financing. The admissibility of $\Phi^z$ then results from
our hypothesis, hence $\Phi^z \in \mathcal{A}_{JD}^Z$. Proposition
\ref{NP-bond-admiss} shows that $\Phi$ is admissible; $\Phi$ is also
V-continuous as consequence of Proposition
\ref{sContinuousStrategiesII}. The NP portfolio $\Phi$ and the
stochastic portfolio $\Phi^z \in \mathcal{A}_{JD}^Z$ are isomorphic,
so according to Corollary \ref{isomophic-v-cont}, $\Phi$ is connected to
$\Phi^z$ and $\Phi \in [\mathcal{A}_{JD}^Z]$.
\end{proof}

Analogously to Proposition \ref{counterex-simple}, simple portfolio strategies
may not be V-continuous on the trajectory space $\mathcal{J}_{\tau}^{\sigma, C}$.

\section{Implications to Stochastic Frameworks}\label{stoch-examples}

In previous sections we have studied some connections between stochastic and
NP-markets, in particular Theorem \ref{main-no-arbitrage} was used
to establish that some NP models are NP-arbitrage free. In this
section we use Theorem \ref{main-no-arbitrageDual} to prove the no
arbitrage property in stochastic settings using some of the results
that we have obtained previously for NP-models. In this way, results
obtained in a non-probabilistic framework not only constitute a
different approach to the main financial problems of hedging and arbitrage but can also be
used  as a technical tool to obtain new arbitrage results in stochastic settings as well.

One important property of this approach for pricing derivatives in probabilistic models is
that it is applicable even in cases where prices are not semimartingales.
From our point of view, this is one a main advantage of the approach, allowing
to price derivatives for  some models where the risk neutral approach is
impossible to carry out or is not very clear.  Another important feature of
this approach is that it encompasses models with jumps and without jumps as we illustrate
in the examples in this section.

\begin{example}[Black-Scholes related models]\label{generalized-BS}
Let $X$ be a price process on $(\Omega, \mathcal{F},
\mathcal{F}_t, P)_{0 \leq t \leq T}$ defined by $X_t =x_{0}~e^{\sigma Z^{Gen}_t}$
where $Z^{Gen}$ is a (general) process adapted to $\mathcal{F}_t$ satisfying\\
1)$[Z^{Gen}]_t^{\tau}=t$.\\
2)$Z^{Gen}$ satisfies a small ball property on $\mathcal{Z}_{\mathcal{T}}([0, T])$ with respect to the
uniform metric.\\
As consequence of 1) Theorem \ref{Non-stoch-BS} applies in this case indicating that hedging is possible
in a path by path sense.\\

\noindent Consider $\mathcal{A} \equiv [\mathcal{A}^Z_{BS}]$ introduced in
Theorem \ref{no-arbitrage-BS} and $ \mathcal{A}^X \equiv
[\mathcal{A}]^X$, where $[\mathcal{A}]^X$ is given as in Definition
\ref{def:connectionII}. Under these circumstances we will argue that
the stochastic market $(X,\mathcal{A}^X)$ on $(\Omega, \mathcal{F},
\mathcal{F}_t, P)$ is arbitrage-free in the classical probabilistic
sense.

\noindent We will apply Theorem \ref{main-no-arbitrageDual} to the
given NP market $(\mathcal{J}_{\tau}^{\sigma}, \mathcal{A})$ with
the supremum metric and  $\mathcal{A}$ as defined above. Condition $C_0$
is an obvious consequence of 1). Condition $C_1$ is derived from 2)
and the fact that $\exp(\cdot)$ (used to construct $X$) is a continuous
function, hence uniformly continuous on compacts.

\noindent As all conditions of Theorem
\ref{main-no-arbitrageDual} are satisfied and
$(\mathcal{J}_{\tau}^{\sigma}, \mathcal{A})$ is NP-arbitrage free by
means of Theorem \ref{no-arbitrage-BS} then the stochastic market
$(X,\mathcal{A}^X)$ is arbitrage free in the classical probabilistic
sense.

\vspace{.05in} \noindent Using similar arguments to the ones we have
used in Section \ref{hedgingAndNoArbitrageInClasses}, it is possible
to verify that the smooth strategies as
previously defined belong to
$\mathcal{A}^X$, this means that $\mathcal{A}^X$ is a large class.
Moreover, also the delta hedging strategies belong to
$\mathcal{A}^X$, as well as simple portfolio strategies like those
in Remark \ref{simple-piece-wise}.

\end{example}

Next we will show that examples of such processes $Z^{Gen}$ are:
\begin{itemize}
\item $Z^F=W+B^H$ where $W$ is a Brownian Motion and $B^H$ is a fractional Brownian
motion  with Hurst index $1/2<H<1$ independent of $W$.
\item $Z^R=\rho W+\sqrt{1-\rho^2}B^R$ where $W$ is a Brownian Motion, $B^R$ is a
reflected Brownian Motion independent of $W$.
and $\rho$ is a real number, $0< \rho <1$.
\item $Z^w$, a Weak Brownian motion.
\end{itemize}

\vspace{.1in}
\noindent \textbf{Mixed Fractional Brownian Model:} For process $Z^F=W+B^H$ we have that
if $1/2<H<1$ then the trajectories of $B^H$ have zero quadratic variation, which implies
that $[Z^F]_t^{\tau}=[W]_t^{\tau}=t$ almost surely. The small ball property of $Z^F$ on
$\mathcal{Z}_{\mathcal{T}}([0, T])$ is consequence of a small ball property of
$W$ on $\mathcal{Z}_{\mathcal{T}}([0, T])$ which is obtained from Lemma
\ref{trajectories-tube}, the hypothesis on independence between
$W$ and $B^H$, and a small ball property of the fractional Brownian motion $B^H$ around the
identically null function (\cite{stolz}, \cite{zahle}).

The absence of arbitrage for this model implies that pricing and hedging
can be done exactly as in the Black-Scholes model. We have proven this fact
using essentially that the trajectories of prices are dense in
$\mathcal{J}_{\tau}^{\sigma}(x_0)$. We did not use any
semi-martingale property of the price process, in fact, for $H \in
(1/2,3/4]$, $X$ is not a semi-martingale, which is a drawback form
the point of view of a risk-neutral approach for pricing.

Some results for models similar to $X_t=x(0)e^{\sigma Z^F_t}$ are presented in \cite{Kloeden} and
\cite{valkeila-1} proving that pricing and hedging procedures in the Black-Scholes
model are robust against perturbations with zero quadratic variation. More recently,
in \cite{valkeila-2}, path-dependents options are replicated under this model.

The replication result in the previous model, which is consequence of $[Z^F]_t^{\tau}=[W]_t^{\tau}=t$
remains true for models satisfying $Z_t=W_t+Y_t$ where $Y$ is a process with zero quadratic variation,
as for example any continuous process with finite variation. While the replication is always valid
for a general $Y$ ,  the no-arbitrage property is less obvious to verify and
will depend on the particular form of $Y$, in the dependence structure between $W$ and $Y$, etc.\\

\vspace{.1in}
\noindent \textbf{Mixed Reflected Brownian Model:}
Consider the process $Z^R=\rho W+\sqrt{1-\rho^2}B^R$ where $W$ is a Brownian Motion and $B^R$ is a
reflected Brownian Motion independent of $W$. By a reflected Brownian motion we understand a process
whose trajectories are obtained by reflecting trajectories of a standard Brownian motion on one or
two reflecting boundaries. A particular example of reflected Brownian motion is given by
$B^R_t=|\tilde{B}_t|$ where $\tilde{B}$ is a Brownian motion. Therefore, our reflected Brownian $B^R$
motion will be upper bounded and/or lower bounded.

In this case we also have $[Z^R]_t^{\tau}=t$ almost surely.
The small ball property of $Z^R$ on $\mathcal{Z}_{\mathcal{T}}([0, T])$ is deduced as follows: we know from Lemma
\ref{trajectories-tube} that $P(\sup_{s \in [0,T]}|W_s-f(s)/\rho|<\epsilon/2 )>0$ for all $f$ continuous such
that $f(0)=0$. On the other hand $P(\sup_{s \in [0,T]}|B^R_s|<\epsilon/2 )>0$
is also true, also as consequence of Lemma
\ref{trajectories-tube}. Then $P(\sup_{s \in
[0,T]}|Z^R(s)-f(s)|<\epsilon )>0$ for all  $\epsilon>0$ and for all
$f$ continuous, with $f(0)=0$ in particular for all $f \in \mathcal{Z}_{\mathcal{T}}([0, T])$.

A financial interpretation of this example is that the asset price
is influenced by an external source of randomness $B^R$ which is
limited within some bounds. This idea has some precedents, see for
example \cite{krugman}. Nevertheless, to our knowledge, this model
has not been considered previously for pricing and hedging purposes.
The risk neutral approach for pricing does not seem to be obvious
for this model, and in fact, could heavily depend on the reflected
boundaries. Our approach to this model (as stated  in Example
\ref{generalized-BS}) is simple and tells us again that pricing and
hedging for this model can be done as in the Black-Scholes paradigm.\\

\vspace{.1in}
\noindent \textbf{Weak Brownian Motion:}
A weak Brownian motion $Z^w$ of order $k \in \mathbb{N}$ is a stochastic process whose
$k$-marginal distributions are the same as of a Brownian motion,
although it is not a Brownian motion. In particular we will
consider those weak Brownian motions of order at least 4 such that their
law on $C[0,T]$ are equivalent to the Wiener measure on $C[0,T]$.
The existence of such processes, as well as some of their properties
are established in \cite{follmer-yor}. In particular we will use
that if $k \geq 4$  $\left\langle Z^w\right\rangle_t^{\tau}=t$ almost
surely. The required small-ball property of $Z^w$ on $\mathcal{Z}_{\mathcal{T}}([0, T])$
is consequence of the equivalence between the law of $Z^w$ and the
Wiener measure on $C[0,T]$.

\noindent
We should remark that a Weak Brownian motion may not be a semimartingale,
therefore, the risk neutral approach for pricing may be impossible for models that include them.
Nonetheless, models including weak Brownian motions have been recently studied in \cite{russo},
through the weaker concept of $\mathcal{A}$-martingale.

\begin{example}[Renewal Process]\label{ex:renewal}
Consider $N^R$ to be a renewal process instead of the usual classical Poisson process:
\[
N^R_t=\sum_{i} 1_{[0,t]}(S_i)
\]
where $S_i$ are considered random jump times such that random variables $S_{i+1}-S_i$,
representing the times between jumps, are positive, independent and identically distributed
with some probability distribution $G(x)$.
Consider the price process given by
\begin{equation}\label{jumps}
X_t=x_{0}e^{\mu t}(1+a)^{N^R_t}
\end{equation}
The trajectories of the process (\ref{jumps}) will be in
$\mathcal{J}^{a,\mu}(x_0)$, so the replicating portfolio in
Proposition \ref{only-jumps} also applies to this case for any
probability distribution $G$. In order to have the no-arbitrage
property, additional assumptions must be made on $G$. For example,
if the support of $G$ is a finite interval it could lead to obvious
arbitrage opportunities related to the imminent occurrence of a
jump. Nevertheless, if  $G$ is absolutely continuous with respect to
the Lebesgue measure on the whole positive real line, the model is
arbitrage-free for a large class of portfolio strategies as the following analysis shows.

\noindent Let $X$ be a price process on $(\Omega, \mathcal{F},
\mathcal{F}_t, P)$ as defined in (\ref{jumps}).
Consider $\mathcal{A} \equiv [\mathcal{A}^Z_{P}]$ introduced in
Theorem \ref{no-arbitrage-jumps}. Let also $ \mathcal{A}^X \equiv
[\mathcal{A}]^X$, where $[\mathcal{A}]^X$ is given as in Definition
\ref{def:connectionII}. Under these circumstances we will argue that
the stochastic market $(X,\mathcal{A}^X)$ on $(\Omega, \mathcal{F},
\mathcal{F}_t, P)$ is arbitrage-free in the classical probabilistic
sense.

\noindent We will apply Theorem \ref{main-no-arbitrageDual} to the
given NP market $(\mathcal{J}^{a, \mu}, \mathcal{A})$ with the
Skorohod's metric and  $\mathcal{A} \equiv [\mathcal{A}^Z_P]$. As we
observed, condition $C_0$ from Theorem \ref{main-no-arbitrageDual} is satisfied. Condition $C_1$,
also from Theorem \ref{main-no-arbitrageDual}, is valid
because of our hypothesis on the support of $G$: the probability of
those trajectories jumping exactly in a small neighborhood of the
jumps of any $x \in \mathcal{J}^{a,\mu}(x_0)$ (which guarantees that
these trajectories are close to $x$ in the Skohorod topology) is
positive. Here implicitly we also used that the times between jumps
are independent random variables. As all
conditions of Theorem \ref{main-no-arbitrageDual} are satisfied and
$(\mathcal{J}_{\tau}^{\sigma}, \mathcal{A})$ is NP-arbitrage free by
means of Theorem \ref{no-arbitrage-BS} then the stochastic market
$(X,\mathcal{A}^X)$ is arbitrage free in the classical probabilistic
sense.

Let $\Phi$  be any of the NP portfolio strategies considered either in Theorem  \ref{only-jumps} or Proposition \ref{sContinuousStrategies}, it then follows from Corollary \ref{isomophic-v-cont}
that  $\Phi^X(t, w) \equiv \Phi(t, X(w))$ belongs to  $\mathcal{A}^X$.
\end{example}

An specific example is when  $G(x)$ is such that $1-G(x) \sim x^{-(1+\beta)}$ with
$\beta \in (0,1)$. This particular case  is used in \cite{valkeila} for the approximation
of a Geometric Fractional Brownian motion.
There, using path-by-path arguments it is shown that model is complete and arbitrage-free for some
class of portfolio strategies.

We should remark that a similar result can be obtained also if the support of $G$ is dense in $[0, \infty)$,
in other words, if every interval $[a,b]$ with $0<a<b$ has positive probability according to $G$: $G(b)>G(a)$,
as suggests the proof of Theorem \ref{no-arbitrage-jumps}. It means that if $G$ is supported on the set of
positive rational numbers $\mathbb{Q}_+$ for example, the no-arbitrage property is valid, thus
pricing can be done exactly as in the Geometric Poisson model. This result, which is analogous in
some way to the one presented in Example \ref{rational-end-points}, is surprising in the sense that the
measure that $G$ induces on $\mathcal{J}^{a,\mu}(x_0)$ is not absolutely continuous with respect to the measure
induced by the Geometric Poisson model.

\begin{example}[Jump-diffusion related models] \label{ex:mixed-jump-diff}

Consider a stochastic process having the form
\begin{equation}\label{eq:general-jd}
X_t=e^{(\mu-\sigma^2/2)t + \sigma Z^G_t}\prod_{i=1}^{N^R_t}(1+X_i),
\end{equation}
where $Z^G$ is a continuous process satisfying that $\langle Z^G\rangle_t=t$. We also
assume that $Z^G$ satisfies a small ball property on $\mathcal{Z}_{\mathcal{T}}([0, T])$
with respect to the uniform norm. Examples of such $Z^G$ are the processes $Z^F$, $Z^R$ and $Z^w$,
previously defined. Process $N^R$ is a renewal process as considered in Example
\ref{ex:renewal} and random variables $X_i$ are considered to be independent with
common distribution $F_X$. Set
$\mathcal{A} \equiv [\mathcal{A}_{JD}^Z]$  and consider the NP market
$(\mathcal{J}_{\tau}^{\sigma, C}, \mathcal{A})$ defined in Theorem
\ref{noArbitrageHedginStrategies}. As we have defined before, let $ \mathcal{A}^X \equiv
[\mathcal{A}]^X$, where $[\mathcal{A}]^X$ is given by (\ref{def:connectionII}).
Also assume $F_X$ and the set $C$ satisfy the assumptions
1) and 2) of Theorem \ref{noArbitrageHedginStrategies}.
Using the same arguments we used there, it is possible to verify that all conditions of Theorem
\ref{main-no-arbitrageDual} are satisfied for the stochastic process $X$ in (\ref{eq:general-jd})
and the NP market $(\mathcal{J}_{\tau}^{\sigma, C}, \mathcal{A})$, therefore the stochastic market
$(X,[\mathcal{A}]^X)$ is arbitrage free in the classical probabilistic sense.
Consider $\Phi$ to be any of the NP portfolio strategies considered in Proposition \ref{sContinuousStrategiesII}, it then follows from Corollary \ref{isomophic-v-cont}
that  $\Phi^X(t, w) \equiv \Phi(t, X(w))$ belongs to  $\mathcal{A}^X$.
\end{example}

The examples introduced above indicate that replication and arbitrage problems in
a stochastic framework could be studied without requiring the semi-martingale property.
Instead, topological properties of the support
of the price process play a key role.
Under these circumstances, the existence of a risk neutral measure constitutes a useful tool for
pricing but not a necessary condition for stating and solving a pricing problem in a coherent way.

\section{Conclusions and Further Work}

The present work develops a non probabilistic framework for pricing using the
classical arguments of hedging and no arbitrage. We obtain general no arbitrage
results in our framework that no longer rely on a probabilistic structure, but
in the topological structure of the space of possible price trajectories and a convenient
restriction on the admissible portfolios to those satisfying certain continuity properties.
Apparently, the introduced non probabilistic framework is far from being a different and isolated approach,
this is strongly suggested  by the fact
that our results in  a non probabilistic setting have also implications on stochastic
frameworks because of the existing connections between both approaches. Therefore, the
results can be used to price derivatives in non standard stochastic models.

There are several possible extensions of our work, in particular, there are many
spaces of trajectories that could be encoded in a non probabilistic framework.
In relation to this we mention that many of the results presented in the paper
can be extended by introducing the analogue of stopping times in our non
probabilistic framework.  
It seems that a major technical advance  for the proposed formalism would be to supply a proof technique
that allows to establish non arbitrage results without relying on the known results for the probabilistic setting.

\appendix

\section{Quadratic Variation and Ito Formula}  \label{pathwiseIto}

From \cite{follmer} we have the following.

\begin{proposition}(It\^{o}-F\"ollmer Formula) Let $x$ be of quadratic variation along $\tau$, and let $y^1, \ldots, y^m$
be continuous functions of bounded variation. Suppose that $f \in C^{1,2,1}([0,T)\times \mathbb{R} \times \mathbb{R}^m )$, then
for all $0\leq s<t<T$:

\begin{equation}          \label{Ito-Follmer}
f(t,x_t,y^1_t,\ldots,y^m_t) = f(s,x_s,y^1_s,\ldots,y^m_s)+\int_s^t \frac{\partial}{\partial t} f(u,x_u,y^1_u,\ldots,y^m_u) du
 +
 \end{equation}
 \begin{equation} \nonumber
 \int_s^t \frac{\partial}{\partial x} f(u,x_u,y^1_u,\ldots,y^m_u) dx_u
  +\frac{1}{2}\int_s^t \frac{\partial^2}{\partial x^2} f(u,x_u,y^1_u,\ldots,y^m_u) d \langle x \rangle_u^{\tau} +
  \end{equation}
 \begin{equation} \nonumber
  \sum_{i=1}^m \int_s^t \frac{\partial}{\partial y^i} f(u,x_u,y^1_u,\ldots,y^m_u) d y^i_u+
\end{equation}
 \begin{equation} \nonumber
 \sum_{u\leq t}\left[ f(u,x_u,y^1_u,\ldots,y^m_u) -  f(u^-,x_{u^-},y^1_{u^-,}\ldots,y^m_{u^-}) \right] -
 \end{equation}
 \begin{equation} \nonumber
  \frac{\partial}{\partial x}f(u^-,x_{u-},y^1_{u-},\ldots,y^m_{u-}) \Delta x_u - \sum_{i=1}^m \frac{\partial}{\partial y^i} f(u^-,x_{u^-},y^1_{u^-,}\ldots,y^m_{u^-}) \Delta y^i_s.
\end{equation}
\end{proposition}

Here are some results that we use in the main body of the paper.

\begin{itemize}

\item If $z(s),~s \in [0,T]$ is a continuous function with zero quadratic variation along
$\tau$ and $x$ is of quadratic variation along $\tau$, then
$[x+z]_t^{\tau}= [x]_t^{\tau}$.

\item If $x$ is of quadratic variation along $\tau$ and $f \in C^1(\mathbb{R})$ then $y=f \circ x$
is of quadratic variation along $\tau$ , moreover
\begin{equation}  \label{quadVarForComposition}
\langle y \rangle_t^{\tau}=\int_0^t (f'(x(s)))^2  d \langle x \rangle_s^{\tau}.
\end{equation}
\end{itemize}
For more details, see \cite{follmer}.

\section{Technical Results}  \label{technicalResults}
\begin{proposition} \label{bond-admiss}
Let $(X_t)_{0 \leq t \leq T}$ be an adapted process on the filtered
probability space $(\Omega, \mathbb{F},(\mathcal{F}_t)_{t \geq 0},
\mathbb{P})$ representing the price process. Assume that a.s. the
trajectories of $X$ belong to some class of functions $\mathcal{J}$
and let $g_1$, $g_2$, $g_m$ be hindsight factors over $\mathcal{J}$.
Let the stochastic portfolio $\Phi=(\phi,\psi)$ have the form
\begin{equation}\label{stock-invest}
\phi_t=G(t,X_{t-},g_1(t,X),g_2(t,X), \ldots ,g_m(t,X))
\end{equation}
and
\begin{equation}\label{bond-invest}
\psi_t=\tilde{V}_{\Phi}(t-)-\phi_t X_{t-}
\end{equation}
with $G \in C^{1}\left([0,T] \times \mathbb{R}\times
\mathbb{R}^m\right)$ ,$g_1(t,X)$, $g_2(t,X)$,..., $g_m(t,X)$ are
$(\mathcal{F}_{t-})$-measurable random variables, and
$\tilde{V}_{\Phi}(t-)=\lim_{s \to t-} \tilde{V}_{\Phi}(s)$ with
\begin{equation}\label{value-predictable}
\tilde{V}_{\Phi}(s)=V_{\Phi}(0)+\int_0^s
G(r,X_{r-},g_1(r,X),g_2(r,X)...,g_m(r,X))dX_r
\end{equation}

Then the stochastic portfolio $\Phi=(\phi,\psi)$ is predictable,
LCRL, and self-financing.
\end{proposition}
\begin{proof}
The held number of stock at time t, $\phi_t$, given by
(\ref{stock-invest}), is predictable (also LCRL) because all $g_i$'s
and $X_{t-}$ are predictable (also LCRL) and $G$ is continuous.
The integral\\
\[
\int_0^s G(r,X_{r-},g_1(r,X),g_2(r,X)...,g_m(r,X))dX_r
\] in (\ref{value-predictable}) can be computed using Ito's formula in terms of
$X_{s-}$ as well as predictable functions (measurable with respect
to $\mathcal{F}_{s-}$) of $X$, which easily implies that
\[
\lim_{s \to t} \int_0^s
G(r,X_{r-},g_1(r,X),g_2(r,X)...,g_m(r,X))dX_r
\]
is measurable with respect to $\mathcal{F}_{t-}$, therefore
$\tilde{V}_{\Phi}(t-)$ is predictable. As $\tilde{V}_{\Phi}(t-)$,
$X_{t-}$, and $\phi_t$ are predictable and LCRL, $\psi_t$ is also
predictable and LCRL from expression (\ref{bond-invest}). Then we
conclude that $\Phi$ is predictable
and LCRL.\\

Let us prove now the self-financing condition. From the definition
of $\tilde{V}_{\Phi}(t)$ and $\tilde{V}_{\Phi}(t-)$ it follows that

\[
\tilde{V}_{\Phi}(t)-\tilde{V}_{\Phi}(t-)=\phi_t(X_t-X_{t-})
\]
hence we have
\[
\tilde{V}_{\Phi}(t)=\tilde{V}_{\Phi}(t-)+\phi_t(X_t-X_{t-})=\phi_t
X_{t-}+\psi_t+\phi_t(X_t-X_{t-})=\phi_t X_{t}+\psi_t
\]
Previous expression means that $\tilde{V}_{\Phi}(t)=V_{\Phi}(t)$,
hence $\Phi$ is self-financing.
\end{proof}

We have also an analogous result in the NP framework.

\begin{proposition} \label{NP-bond-admiss}
Consider a class of trajectories $\mathcal{J}$ and let $g_1$, $g_2$,
$g_m$ be hindsight factors over $\mathcal{J}$ such that all $g_i$
satisfy the additional condition:
\begin{equation}\label{NP-pred-g}
g_i(t,x)=g_i(t,\tilde{x})
\end{equation}
whenever $x(s)=\tilde{x}(s)$ for all $0\leq s <t$\\

For all $x \in \mathcal{J}$ let the portfolio $\Phi=(\phi,\psi)$
have the form
\begin{equation}\nonumber 
\phi(t,x)=G(t,x(t-),g_1(t,x),g_2(t,x)...,g_m(t,x))
\end{equation}
and
\begin{equation}   \label{NP-bond-invest}
\psi(t,x)=\tilde{V}_{\Phi}(t-,x)-\phi(t,x) x(t-)
\end{equation}
with $G \in C^{1}\left([0,T] \times \mathbb{R}\times
\mathbb{R}^m\right)$ and $\tilde{V}_{\Phi}(t-,x)=\lim_{s \to t-}
\tilde{V}_{\Phi}(s,x)$ with
\begin{equation} \label{value-predictable}
\tilde{V}_{\Phi}(s,x)=V_{\Phi}(0,x_0)+\int_0^s
G(r,x(r-),g_1(r,x),g_2(r,x)...,g_m(r,x))dx_r
\end{equation}

Then the NP portfolio $\Phi=(\phi,\psi)$ is NP-predictable, LCRL,
and NP-self-financing.
\end{proposition}
\begin{proof}
The NP-predictability of $\phi$ is an obvious consequence of
(\ref{NP-pred-g}). The NP-predictable representation of $\int_0^s
G(r,x(r-),g_1(r,x),g_2(r,x)...,g_m(r,x))dx_r$, which is given by
Ito-F\"ollmer formula, guarantees that also $\psi$ is NP-predictable,
therefore portfolio  $\Phi=(\phi,\psi)$ is NP predictable.

The self-financing property and the LCRL property can be proved
exactly as in Proposition \ref{bond-admiss}.
\end{proof}

\begin{lemma}\label{lemma:neighboors}
 Let $ x^* \in \mathcal{J}_{\tau}^{\sigma,C}(x_0)$ and let $\left\{x^{(n)} \right\}_{n=0,1...}$  with
$ x^{(n)} \in \mathcal{J}_{\tau}^{\sigma,C}(x_0)$
be a sequence of functions converging to $x^*$ on the Skorohod's  topology. Then there exists $M \in \mathbb{N}$
such that if $n > M$ then $ x^{(n)}$ has the same number of jumps as $x^*$. Moreover, if the jump times of $x^*$
are denoted as $s_1,s_2,..., s_m$ and the jump times of $ x^{(n)}$ (for $n > M$) are denoted as
$s_1^{(n)},s_2^{(n)},..., s_m^{(n)}$, then for all $i=1,2...,m$:
\begin{equation} \nonumber
s_i^{(n)} \to s_i
\end{equation}
and
\begin{equation}\label{eq:jump-siz}
\left[x^{(n)}(s_i^{(n)})-x^{(n)}(s_i^{(n)}-)\right] \to \left[x^*(s_i)-x^*(s_i-)\right]
\end{equation}
\end{lemma}

\begin{proof}As $x^*$ is c\`ad-l\`ag on $[0,T]$,
$x^*$ attains its maximum and minimum on this interval. On the other hand $x^*(t)>0$ for all $0 \leq t \leq T$,
so $\inf_{t \in [0,T]} x^*(t)=\min_{t \in [0,T]} x^*(t)=x^*_{min}>0$.

If $s$ is a jump time for $x^*$, $|x^*(s)-x^*(s-)| > x^*(s-) h \geq x^*_{min}h$, which means that the absolute size
of the jumps of $x^*$ are strictly bounded below by $x^*_{min}h$.

Let $\epsilon< \frac{x^*_{min}h}{2(h+2)}$, then there exists $M(\epsilon) \in \mathbb{N}$ such that
$d(x^*, x^{(n)})< \epsilon$ for all $n \geq M(\epsilon)$. As the Skorohod's distance between $x^*$ and $x^{(n)}$,
$d(x^*, x^{(n)})< \epsilon$, there exist an increasing function $\lambda_n: [0,T] \rightarrow [0,T]$
with $\lambda_n(0)=0$, $\lambda_n(1)=1$, $|\lambda_n(t)-t|<\epsilon$ for $0 \leq t \leq T$ and

\begin{equation}\label{uniform-epsilon}
\left|x^*(t)-x^{(n)}(\lambda_n(t)) \right| < \epsilon, \; \text{for} \; 0 \leq t \leq T.
\end{equation}

Let $s$ be a jump time of $x^*$. As expression (\ref{uniform-epsilon}) is valid for $0 \leq t \leq T$,
for any increasing sequence of positive real numbers $\{s_i\}_{i=1,2,...}$ converging to $s$ we have that:

\begin{equation}\label{uniform-sequence}
\left|x^*(s_i)-x^{(n)}(\lambda_n(s_i)) \right| < \epsilon, \; \text{for all} \; i.
\end{equation}

As $x^*$ and $x^{(n)}$ are c\`ad-l\`ag function on $[0,T]$ it is possible to take limits in
(\ref{uniform-sequence}) as $i \to \infty$ obtaining:

\begin{equation}\label{uniform-limit}
\left|x^*(s-)-x^{(n)}(\lambda_n(s)-) \right| < \epsilon.
\end{equation}

From expressions (\ref{uniform-epsilon}) and (\ref{uniform-limit}) we have

\begin{equation}\label{jumps-limit}
\left| \left(x^*(s)-x^*(s-)\right)-\left(x^{(n)}(\lambda_n(s))-x^{(n)}(\lambda_n(s)-)\right) \right| < 2\epsilon,
\end{equation}

therefore

\begin{eqnarray*}
\left| x^{(n)}(\lambda_n(s))-x^{(n)}(\lambda_n(s)-) \right| &>& \left| x^*(s)-x^*(s-)\right| -2\epsilon\\
                                                            &>& x^*_{min}h-\frac{2x^*_{min}h}{2(h+2)}\\
                                                            &=& x^*_{min}h  \frac{h+1}{h+2}>0.
\end{eqnarray*}
This means that if $s$ is a jump time for $x^*$, then $\lambda_n(s)$ is also a jump time
for $x^{(n)}$. \\

Consider now that $s'$ is a jump time of $x^{(n)}$ and define $s''=\lambda_n^{-1}(s')$.
We have from (\ref{jumps-limit}) that

\begin{equation} \nonumber 
\left| \left(x^*(s'')-x^*(s''-)\right)-\left(x^{(n)}(s')-x^{(n)}(s'-)\right) \right| < 2\epsilon.
\end{equation}

Thus,

\begin{equation}\label{inverted-inequality}
\left| x^*(s'')-x^*(s''-) \right| > \left| x^{(n)}(s')-x^{(n)}(s'-)\right| - 2\epsilon.
\end{equation}
As $d(x^{(n)},x^*)<\epsilon$ then $\inf_{s \in [0,T]}(x^{(n)}(s))>x^*_{min}-\epsilon>0$ uniformly on n.
Therefore, (\ref{inverted-inequality}) becomes:

\begin{equation*} \nonumber
\left| x^*(s'')-x^*(s''-) \right| > (x^*_{min}-\epsilon)h - 2\epsilon=\frac{x^*_{min}h}{2}>0,
\end{equation*}
which implies that $s''$ is a jump time for $x$ if $s'=\lambda_n(s'')$ is a jump time for $x^{(n)}$.

Previous analysis tells that for $n$ large enough, $x^{(n)}$ has exactly the same number of
jumps as $x^*$ has.

In order to derive (\ref{eq:jump-siz}), consider $\Delta_i x^*$ as the size of the i-th jump of
$x^*$ and $\Delta_i x^{(n)}$ as the size of the i-th jump of $x^{(n)}$. Then:
\begin{eqnarray*} \nonumber
\Delta_i x^{(n)}&=&\left(x^{(n)}(s_i^{(n)})-x^{(n)}(s_i^{(n)}-)\right)\\
                &=&\left(x^{(n)}(\lambda_n(s_i))-x^{(n)}(\lambda_n(s_i)-) \right).
\end{eqnarray*}
Expression (\ref{jumps-limit}) implies that:
\begin{equation} \nonumber
\left(x^{(n)}(\lambda_n(s_i))-x^{(n)}(\lambda_n(s_i)-) \right) \to \left( x^*(s_i)-x^*(s_i-) \right),
\end{equation}
as $n \to \infty$, so  $\Delta_i x^{(n)} \to \Delta_i x$ as $n \to \infty$.
\end{proof}

\end{document}